\documentclass[12pt,A4,fleqn]{article}    
\usepackage{epsfig}
\usepackage{amssymb,amsmath,amsthm}
\usepackage[title]{appendix}
\usepackage{tcolorbox}
\usepackage{caption}
\usepackage{tikz}
\usetikzlibrary{intersections}
\usepackage{bm}
\usepackage{rotating}
\usepackage{multirow}
\usepackage{hyperref}
\usepackage{float}
\usepackage{enumerate}
\usepackage{parskip}
\usepackage{natbib}
\usepackage{algorithm}
\usepackage{algpseudocode}
\usepackage{diagbox}
\usepackage{threeparttable}

\numberwithin{table}{section}
\numberwithin{equation}{section}
\numberwithin{figure}{section}
\usepackage{sectsty}
\subsectionfont{\small\selectfont}
\sectionfont{\normalsize \selectfont}
\subsubsectionfont{\footnotesize \selectfont}

\newtheorem{prop}{Proposition}[section]
\newtheorem{theorem}{Theorem}[section]
\newtheorem{condition}{Condition}
\newtheorem{corollary}{Corollary}[section]
\newtheorem{lemma}[theorem]{Lemma}

\DeclareMathOperator{\E}{\mbox{E}}

\title{Optimal reinsurance for risk over surplus ratios}
\author{Erik B\o lviken and Yinzhi Wang\\
	Department of Mathematics\\
	University of Oslo} 
\date{\today}

\begin{document}
	\maketitle
	\noindent
	\small 
	\newpage 
	{\bf Abstract}
	\\\\
Optimal reinsurance when Value at Risk and expected
surplus is balanced through their ratio is studied, and it is demonstrated
how results for
risk-adjusted surplus can be utilized. Simplifications for large
portfolios are derived, and this large-portfolio study suggests a new condition
on the reinsurance pricing regime which is crucial for the results
obtained. 
One or two layer contracts now become optimal
for both risk-adjusted surplus and the risk over expected surplus ratio, but
there is no second layer when portfolios are large or when 
reinsurance 
prices are below some threshold.
Simple 
approximations of the optimum portfolio is considered, and 
their degree of
degradation compared to the optimum is studied which leads to  
theoretical degradation rates  
as the number of policies grow. 
The theory is supported by numerical experiments which suggest
that the shape of the claim severity distributions may not  be of primary importance
when designing an optimal reinsurance program.
It is argued that the approach
can be applied to
Conditional
Value at Risk as well.
	\\\\\\\\\\\\\\\\\\\\\\\\\\\\\\\\\\\\\\\\\\\\\\\\\\
	{\bf Key words and phases}
	\\\\
	Asymptotics, degradation rates, large portfolios, one- and two-layer 
	contracts, reinsurance pricing regimes, risk-adjusted surplus.
	\newpage
\section{Introduction}\label{sec:intro}
Actuarial literature contains 
countless formulations of what optimal reinsurance should mean,
for example, \cite{borch1960, arrow1963, kaluszka2001} and \cite{cheung2014optimal}. A criterion with much sense industrially
is to balance risk and profit through a ratio where a risk measure is
divided on expected surplus. Most of the paper makes use of  Value at Risk
in this role.
That is how the insurance industry is regulated at present, but we shall argue
that the perhaps theoretically more appealing Conditional Value at Risk
could be handled in the same manner which would yield similar, 
but not quite identical 
results; see the companion paper \cite{wang2019}.
Ratios of risk and expected surplus are related to
risk-adjusted surplus under which many authors have examined how
reinsurance could be optimized, see \cite{chi2012reinsurance,asimit2013optimal,cheung2017characterizations} and \cite{chi2017optimal}. Of particular 
significance for the present paper is 
\cite{chi2017optimal} who were able to show 
that the
reinsurance treaties maximizing risk-adjusted surplus are of the
multi-layer type. It is not tenable 
to assume that reinsurance pricing is based on a fixed loading,
as noted already 
by \cite{borch1960}, and \cite{chi2017optimal} made use of 
a general formulation of reinsurance pricing
that goes back to \cite{buhlmann1980economic}. Many other
researchers, for
example, \cite{chi2013optimal} and \cite{zhuang2016marginal} have used 
this scheme. 
\\\\
We follow in their track except that we argue for a modification. 
The world of reinsurance is above all a market with
pricing offers from reinsurers
defining a supply curve for such risk, but
the variation in time is enormous with big events
like the World Trade Center terrorism in $2001$ 
having huge impact.
Details are not open to
the public in any case, and academic studies must therefore employ  
`premium principles' as proxies for the real market prices. The
question is what conditions should be imposed on them. It is demonstrated
in this paper
that the original set-up in \cite{buhlmann1980economic} would enable the insurer
to reinsurance everything up to Value at Risk and for large portfolios
still obtain profit. No net 
solvency capital would then be necessary, but is it likely 
that the insurance 
market should allow such a situation to exist? We think not and have
derived from this viewpoint a new condition that differs from
the one in current use. 
All our results
depend on it. One of the consequences is that the multi-layer solution
for risk-adjusted surplus 
in \cite{chi2017optimal} is reduced to one or two layers with more than one 
layer only when risk is expensive, and this result is a stepping stone
to similar results for the risk over surplus ratio.
\\\\ 
One issue that does not seem to have been treated
in actuarial literature is what happens
when the portfolio size become infinite.
This is highly relevant for reinsurance of single risks
since many of them are sums of a large number of policies. 
Such asymptotic studies have in 
statistics or other branches of applied mathematics often
lead to simplification and clarity as indeed they do here.
We have limited ourselves to independent risks so that we can
lean on the central limit theorem 
and its Lindeberg extension, but the approach can without doubt be extended
to dependent risks influenced by some common random factor; more on that
in the concluding section. 
A part of all asymptotic studies is how well such approximations
perform for finite portfolio sizes. This is not in our case a question of error,
but rather one of degradation of the criterion in terms of how much
it changes compared to its value at the strict optimum.
Theoretical degradation studies are developed  
in Section \ref{sec:large_port} through large portfolio analysis
which leads to approximate rates of decay as the number of policies grows. 
The numerical side is examined in Section \ref{sec:numerical} 
though simulation studies and illustrates how large portfolios must
be for the large-portfolio approximations to perform well.

\section{Basics and preliminaries}\label{sec:basics}
\subsection{Notation and formulation}\label{subsec:notations}
Let $X$ be the total claim losses of a single portfolio
of non-life insurance policies over a certain period of time 
(often one year) and let $R_I=R_I(X)$ be the net risk retained after 
having received
from a reinsurer $I=I(X)$ so that $R_I(X)=X-I(X)$.
Natural restrictions on $I(X)$ are
\begin{equation}
0\leq I(X)\leq X
\hspace*{1cm}\mbox{and}\hspace*{1cm}
0\leq I(X_2)-I(X_1)\leq X_2-X_1
\hspace*{0.2cm}\mbox{if}\hspace*{0.2cm}X_1\leq X_2
\label{e21}
\end{equation}
where the first condition is obvious since the reinsurer will never 
pay out more than the original claim. The second condition 
is known as the slow growth property, and it opens for moral hazard
if it isn't satisfied; consult \cite{cheung2014optimal}. Contracts 
satisfying~(\ref{e21}) will be referred to as feasible.

Let $F(x)$ be the distribution function of $X$ which 
starts at the origin. Risk measures
that has attracted much interest are 
Value at Risk and Conditional Value at Risk with
formal mathematical definitions 
\begin{equation}
\mbox{VaR}_\epsilon(X)=\inf\{x|1-F(x)\leq \epsilon\}
\hspace*{1cm}\mbox{and}\hspace*{1cm}
\mbox{CVaR}_\epsilon(X)=E\{X|X\geq\mbox{VaR}_\epsilon(X)\}
\label{e22}
\end{equation}
where $\epsilon>0$ is a given level. 
When these quantities apply to the retained risk $R_I$
we shall be using notation like $\mbox{VaR}_\epsilon(R_I)$. 
The right inequality in~(\ref{e21})
implies that the retained risk $R_I(X)$ is non-decreasing in $X$ so that
if $x_\epsilon=\mbox{VaR}_\epsilon(X)$, then
\begin{equation}
\mbox{VaR}_\epsilon(R_I)=R_I(x_\epsilon)
\hspace*{1cm}\mbox{and also}\hspace*{1cm}
I(x_\epsilon)=x_\epsilon-R_I(x_\epsilon).
\label{e23}
\end{equation}
Premia involved are $\pi$ collected by the cedent from its customers
and $\pi(I)$ for the reinsurance. Those are in their simplest form 
\begin{equation}
\pi=(1+\gamma)E(X)
\hspace*{1cm}\mbox{and}\hspace*{1cm}
\pi(I)=(1+\gamma^{\mbox{\tiny re}})E\{I(X)\}
\label{e24}
\end{equation}
with $\gamma>0$ and $\gamma^{\mbox{\tiny re}}>0$ 
given loadings (or coefficients) and
with $\gamma^{\mbox{\tiny re}}>\gamma$ in practice. 
The reinsurance part is inadequate. 
Prices in that  market is likely to increase with risk
beyond the fixed 
coefficient $\gamma^{\mbox{\tiny re}}$ in~(\ref{e24}) right. 
In real life that would be 
captured by offers the cedent company receives 
from reinsurers, but such information isn't available and for academic work 
we must instead supply a so-called premium principle 
which is dealt with in Section \ref{subsec:re_pricing}.
\subsection{Expected surplus}\label{subsec:ex_surplus}
We need a mathematical expression for the expected surplus 
for an insurer
under a given
reinsurance treaty. Into the account goes the
premium $\pi$ collected from clients and out of it
their net 
claims $R_I(X)$ and the reinsurer
premium $\pi(I)$. If the cost of
holding solvency capital is subtracted too, a simplified
summary of the balance sheet becomes
\begin{equation}
{\cal A}(I)=\pi-R_I(X)-\pi(I)-\beta \mbox{VaR}_\epsilon(R_I)
\label{e25}
\end{equation}
with the notation highlighting the dependence on the reinsurance 
function $I(X)$. The
coefficient
$\beta\geq 0$ applies
per money unit. 
It seems industrially plausible 
to attach cost to the entire solvency capital, not only to the 
part above the average,
as in \cite{chi2017optimal}. Alternatively
such cost
might be in terms of the Conditional Value at Risk
with CVaR$_\epsilon$ replacing VaR$_\epsilon$ on the right in~(\ref{e25}).
Let ${\cal G}(I)=E\{{\cal A}(I)\}$ 
and take expectations in~(\ref{e25}). 
Inserting~(\ref{e23}) and~(\ref{e24}) left yield
\begin{equation}
{\cal G}(I)=\gamma E(X)-\beta x_\epsilon-(\pi(I)-E\{I(X)\})+\beta I(x_\epsilon)
\label{e26}
\end{equation}
with the last two terms
depending on the reinsurance contract.
\subsection{Reinsurance pricing}\label{subsec:re_pricing}
A standard formulation, used for example 
in \cite{chi2017optimal}, is to introduce a market factor
$M(Z)$ so that
\begin{equation}
\pi(I)=E\{I(X)M(Z)\}
\label{e27}
\end{equation}
where $Z$ is a positive random variable correlated with $X$ 
and $M(\cdot)$ some function for which an exponential one will be used 
with the examples in Section \ref{subsec:num1}. 
General models for $(X,Z)$ is constructed through copulas.
With $G(z)$ the distribution function of
$Z$ let 
\begin{equation}
X=F^{-1}(U)\hspace*{1cm}\mbox{and}\hspace*{1cm}Z=G^{-1}(V)
\label{e28}
\end{equation}
with $F^{-1}(u)$ and $G^{-1}(v)$ the percentile functions
of $F(x)$ and $G(z)$ and with $(U,V)$ a dependent pair of uniform variables.
This yields an alternative expression for $\pi_I$.
By the rule of double expectation  
\begin{displaymath}
\pi(I)=E\{I(X)M(Z)\}=E\{E\{I(X)M(Z)|U\}\}=E\{I(X)E\{M(Z)|U\}\},
\end{displaymath}
with the last identity due to $X=F^{-1}(U)$ being fixed by $U$.
Hence
\begin{equation}
\pi(I)=E\{I(X)W\{F(X)\}\}
\hspace*{1cm}\mbox{where}\hspace*{1cm}W(u)=E\{M(Z)|u\}.
\label{e29}
\end{equation}
The impact of the dependency between  $X$ to $Z$ is taken care of by $W(u)$ 
where the distribution
function $F(x)$ of $X$ doesn't enter, and this will prove convenient 
when $F(x)$ depends on the underlying portfolio size in Section \ref{sec:large_port}.
\\\\
It is often  assumed that $E\{M(Z)\}=1$,  
but that is hardly an obvious assumption.
It is being violated 
when $M(Z)=1+\gamma^{\mathrm{re}}$ as in~(\ref{e24}) right, and there is 
in the present work no point in restricting the set-up so strongly, more
on that later. Note in passing
that if $W(u)$ is a non-decreasing 
function of $u$, as
is plausible and assumed below, then
\begin{displaymath}
\pi(I)=E\{I(X)W\{F(X)\})\}\geq E\{I(X)\}E\{W\{F(X)\}\}
=E\{I(X)\}E\{M(Z)\}
\end{displaymath}
since $U=F(X)$ and $E\{W(U)\}=E\{M(Z)\}$ by~(\ref{e29}) right. Hence
a non-decreasing $W(u)$ guarantees
the reinsurance premium to be larger than the expected reinsurance 
pay-out if
$E\{M(Z)\}\geq 1$. 
Much more general formulations of premium principles
can be found in \cite{Furman2009}.
\\\\
The price on reinsurance can also be expressed through the function
\begin{equation}
K(u)=\int_{u}^1\{W(v)-1\}dv,\hspace*{1cm}0\leq u\leq 1
\label{e210}
\end{equation}
which is sketched in Figure 2.1 below.
Consider the
reinsurer
expected surplus which by~(\ref{e29}) left is
\begin{displaymath}
\pi(I)-E\{I(X)\}=\int_0^\infty(W\{F(x)\}-1)I(x)dF(x)
\end{displaymath}
or since the derivative $K'(u)=-\{W(u)-1\}$, 
\begin{displaymath}
\pi(I)-E\{I(X)\}
=-\int_0^\infty K'\{F(x)\}I(x)dF(x).
\end{displaymath}
But if ${\cal I}_{x>t}=0$ if $x\leq t$ and $=1$ otherwise.
then
\begin{displaymath}
I(x)=\int_0^x{\cal I}_{x>t}dI(t)
\end{displaymath}
so that
\begin{displaymath}
\pi(I)-E\{I(X)\}
=
-\int_0^\infty K'\{F(x)\} \int_0^x{\cal I}_{x>t}dI(t)dF(x)
=-\int_0^\infty \int_t^\infty K'\{F(x)\} dF(x)dI(t),
\end{displaymath}
after changing the order of integration.
Since $K\{F(x)\}\rightarrow K(1)=0$ as $x\rightarrow \infty$, 
it follows that
\begin{equation}
\pi(I)-E\{I(X)\}=\int_0^\infty K\{F(t)\} dI(t).
\label{e211}
\end{equation}
\subsection{Conditions on  reinsurance pricing}\label{subsec:rein_condi}
The function $K(u)$ will play a key role, and it is possible to extract
some useful properties of it if some restrictions are imposed on the
reinsurance pricing regime, notably:
\\
\begin{condition}\label{condi1}
	(i) $W(u)$ is non-decreasing and (ii)
$E\{M(Z)\}> 1+\gamma$.
\end{condition}
The first assumption assumes a positive type of dependence  
between the market factor $M(Z)$ and the risk $X$, surely reasonable.
A sufficient condition for that is $M(z)$ 
being non-decreasing in $z$ and the model
for $(X,Z)$
in~(\ref{e28})  based on a positive dependent copula for
$(U,V)$ 
in the sense that Pr$(V>v|u)$ is increasing in $u$ 
for all $v$. Most copulas satisfy this, and it
is under these circumstances a trivial matter to verify that 
$W(u)$ is monotone upwards. 
\\\\
The second assumption which will be needed 
for the large-portfolio study in Section \ref{sec:large_port},
may seem less obvious since
many authors assume $E\{M(Z)\}=1$. This
condition goes 
back to \cite{buhlmann1980economic} who derived it through
an economic  equilibrium argument. Section \ref{subsec:4.2} will present
alternative reasoning with some resemblance  to arbitrage 
which leads to Assumption (ii).
From~(\ref{e210})
\begin{displaymath}
K(0)=\int_0^1(W(v)-1)dv=E\{M(Z)\}-1
\end{displaymath}
and the assumption is the same as
\begin{equation}
K(0)>\gamma
\label{e211a}
\end{equation}
which is the version that will be cited below. 
To see what it mean suppose
there is a fixed loading $\gamma^{\mbox{\tiny re}}$ in the reinsurance market
so that
$W(u)=1+\gamma^{\mbox{\tiny re}}$. Then
$K(0)=\gamma^{\mbox{\tiny re}}$, 
and Assumption (ii) implies
$\gamma^{\mbox{\tiny re}}> \gamma$ with the  loading
in the reinsurance market the larger one.
\\\\
Some useful deductions on the form of the function 
$K(u)$ can be drawn under Condition \ref{condi1}:
\\
\begin{lemma}\label{lem2.1}
	Suppose Condition \ref{condi1} is true. Then
	$K(u)\geq 0$ {\em everywhere},
	and 
		there is a unique real number $\delta$  between
	$0$ and $1$
	so that
	\begin{equation}
	K(1-\delta)=\gamma
	\qquad\mbox{and}\qquad K'(1-\delta)\leq 0
	\label{e211b}
	\end{equation}
	where $K'(u)$  is the derivative. 
\end{lemma}

\begin{proof}
	Note that $K'(u)=-W(u)+1$ so that either $K(u)$ is 
	decreasing everywhere or, as in Figure 2.1,  
there is an $u_0$ between $0$ and $1$ so that 
	$W(u_0)=1$ which means that $K(u)$ decreases to the left 
and increases to 
	the right of $u_0$. There is in either case
	a unique $\delta$ as in~(\ref{e211b}) with the derivative of $K(u)$ negative at 
	that point. That $K(u)\geq 0$ is immediate 
when $u>u_0$ since the integrand in~(\ref{e210}) is positive (or zero) 
everywhere whereas we also have
	\begin{displaymath}
	K(u)=K(0)-\int_0^u(W(v)-1)dv
	\end{displaymath}
	which is $\geq 0$ when $u<u_0$ since the 
integrand now is negative.
\end{proof} 
\begin{figure}
	\centering
	\begin{tikzpicture}[scale=1.5]
	\draw[->](-0.3,0) -- (6.2,0) node[below] {$u$};
	\draw[->](0,-0.3) -- (0,4.3) node[left] { $K(u)$};
	\node[below left]at(0,0) {\footnotesize $0$};
	\draw[dashed] (5.3,2)--(0,2) node[left]{$\gamma$};
	\draw[black,line width=1pt,name path=plot](0,2.5)..controls(1.5,4.5) and(2,2.5)..(3.0,1.5)..
	controls (4,0.3) and (4.7,0.2)..(5.9,0);
	\node[below](a)at(1.15,0){$u_0$};
	\path[name path=froma](a)--+(0,4);
	\draw[dashed,name intersections={of=froma and plot}](a)--(intersection-1)--(0,0|-intersection-1)node[left]{\footnotesize $ K(u_0)$};
	\draw[dashed] (2.56,2.0)--(2.56,0) node[below] {$1-\delta$};
	\draw[dashed] (5.3,0.1)--(5.3,0) node[below]{$1-\epsilon$};
	\draw[dashed] (5.9,0)--(5.9,0) node[below]{1};
	
	\end{tikzpicture}
	\caption{Plot of the K function $K(u)$ when it has a maximum
and $\delta>\epsilon$.}\label{fig:2.1}
\end{figure}
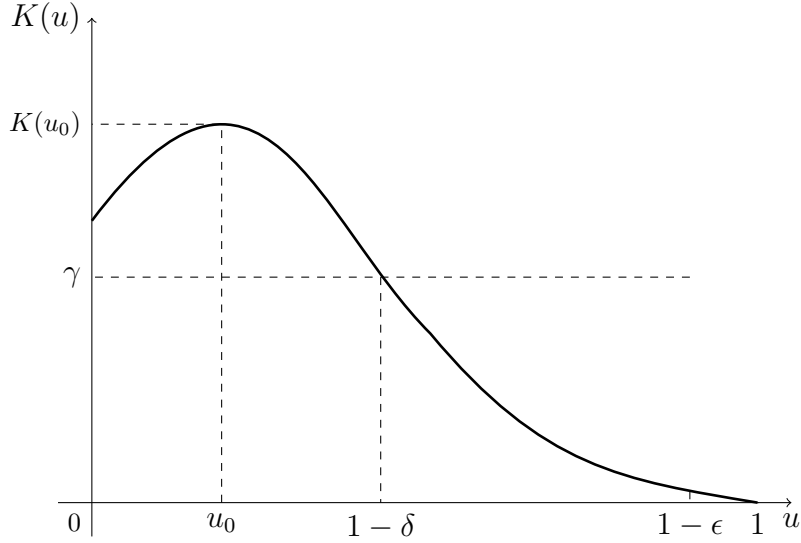

The function $K(u)$ has been plotted in Figure \ref{fig:2.1} when it 
has a maximum.
Its values at $1-\delta$ where it crosses the 
$\gamma$-line and its value at the Value at Risk level $1-\epsilon$ 
will in Section \ref{sec:large_port}
play a main role in defining optimal or nearly optimal reinsurance 
for large portfolios with
the $1-\delta$ and $1-\epsilon$ percentiles of 
the underlying risk variable $X$ being the lower and upper limit
of one-layer contracts. It could happen that the $\gamma$-line crossing in 
 Figure \ref{fig:2.1} takes place to the right of $1-\epsilon$. If so,
the upper and lower limits coincide, and reinsurance is so expensive
that it is optimal for the insurer to carry all risk himself.
It will be of some importance
that $K'(1-\delta)\leq 0$, and in practice this inequality is sharp
which is taken for granted in Proposition \ref{prop4.2} below.

\section{Optimization}\label{sec:opti}
\subsection{Risk-adjusted surplus}\label{subsec:risk_surplus}
Many contributors to reinsurance optimum theory work with
Lagrangian set-ups of the form
\begin{equation}
{\cal L}(I)={\cal G}(I)-\lambda\rho(I)
\label{e212}
\end{equation}
with $\rho(I)$ a risk measure and 
$\lambda>0$ a price on risk;
see \cite{balbas2009optimal,tan2011optimality,Jiang_2017} and \cite{Weng2017}. This
risk-adjusted, expected surplus ${\cal L}(I)$ of the insurer is then maximized,
and as $\lambda$ varies the solutions 
define
an efficient frontier tracing out the  minimum 
$\rho(I)$ obtainable for a given value of 
${\cal G}(I)$. These solutions are
needed when the risk over surplus ratio is studied in Section \ref{subsec:3.3}.\
\\
\begin{prop}\label{prop1}
	\cite{chi2017optimal}. If $\rho(I)=\mbox{VaR}_\epsilon(R_I)$, is Value at Risk,
	then the optimal reinsurance function
	in~(\ref{e212}) is
	\begin{equation}
	I_\lambda(x)=\int_0^x{\cal I}_{\psi_\lambda(y)>0}\,dy
	\hspace*{1cm}\mbox{where}\hspace*{1cm}\psi_\lambda(x)=-K\{F(x)\}+
	(\beta+\lambda){\cal I}_{x<x_\epsilon}
	\label{e213}
	\end{equation}
	with ${\cal I}_A$ the indicator function of the event $A$.
\end{prop}
\begin{proof}
	It isn't assumed that $E\{M(Z)\}=1$ as 
	in \cite{chi2017optimal}, but their ingenious argument
	still works. 
	Insert~(\ref{e26}) and
	$\rho(I)=\mbox{VaR}_\epsilon(R)=x_\epsilon-I(x_\epsilon)$ into~(\ref{e212}). This
	yields
	\begin{displaymath}
	{\cal L}(I)=\gamma E(X)-(\beta+\lambda)x_\epsilon-(\pi_I-E\{I(X)\})+
	(\beta+\lambda) I(x_\epsilon)
	\end{displaymath}
	which can be combined with~(\ref{e211}) for the next to last term 
	on the right and also
	\begin{displaymath}
	I(x_\epsilon)=\int_0^\infty {\cal I}_{x<x_\epsilon}dI(x).
	\end{displaymath}
	for the last one. Hence
	\begin{equation}
	{\cal L}(I)=\gamma E(X)-(\beta+\lambda)x_\epsilon+\int_0^\infty\psi_\lambda(x)dI(x)
	\label{e213a}
	\end{equation}
	with $\psi_\lambda(x)$ as in~(\ref{e213}) right.
	The restrictions in~(\ref{e21}) means that $0\leq dI(x)/dx\leq 1$
	so that ${\cal L}(I)$ is maximized by selecting $dI(x)/dx=1$ whenever
	$\psi_\lambda(x)>0$ and   $dI(x)/dx=0$ otherwise, and this yields
	$I_\lambda(x)$ in~(\ref{e213}) as the optimum.
\end{proof}
\noindent
The proposition shows that the optimum 
is of the  multi-layer type. Let
\begin{equation}
I_{\bf a}(X)=\max(X-a_1,0)-\max(X-a_2,0)
\label{e214}
\end{equation}
where ${\bf a}=(a_1,a_2)^T$ 
is a vector of coefficients 
for which $a_1\leq a_2$. This is a single-layer contract
which is the solution in Proposition \ref{prop1} 
when $\psi_\lambda(x)>0$ between $a_1$ and $a_2$ and $\leq 0$ elsewhere.
The general solution
depends on how many times $\psi_\lambda(t)$ in~(\ref{e213})
crosses zero. If there are
three or four crossings there is an additional layer $I_{\bf b}(x)$
with ${\bf b}=(b_1,b_2)$, and the optimum is now
$I_{\bf ab}(x)=I_{\bf a}(x)+I_{\bf b}(x)$. In theory we may continue
to five or six crossings and a third layer and so on,
but arguably there are in practice at most two:
\\
\begin{corollary}\label{coro1}
If $\lambda\leq K(0)-\beta$, then
the optimum reinsurance function in Proposition \ref{prop1}
is under Assumption (i) in Condition \ref{condi1}
a single layer contract $I_{\bf a}(x)$ with
$a_2=x_\epsilon$. In the opposite case where  $\lambda> K(0)-\beta$, 
there may be
an additional layer $I_{\bf b}(x)$  with $b_1=0$  and $b_2<a_1$. 	
\end{corollary}
\noindent
Degenerate situations are covered by this, for example
$a_1=a_2$ leading to 
no reinsurance at all or $a_1=0$  with the reinsurer covering
everything up to $x_\epsilon$. The (perhaps surprising) $b$-layer starting 
at $0$ occurs when $\lambda$ is large enough; i.e. when
the cost attached $\rho(I)$ weights heavily enough compared to 
the expected
gain ${\cal G}(I)$. 
\begin{proof}
Note that 
if $x>x_\epsilon$, then 
$\psi_\lambda(x)=-K\{F(x)\}\leq 0$
so that there is under no circumstances reinsurance above $x_\epsilon$.
On the other hand  
$\psi_\lambda(0)=-K(0)+\beta+\lambda$ and $\psi_\lambda(x)$ starts below or at
$0$ if $\lambda\leq K(0)-\beta$.
Since the derivative $K'(u)=-\{W(u)-1\}$, Assumption (i) tells us that
either $K\{F(x)\}$ decreases everywhere or start to increase and then decrease.
In either case $\psi_\lambda(x)$ crosses zero from negative to positive
at a single point or not at all.
If $\psi_\lambda(0)>0$ there will be a layer starting at zero and 
a second one if
there are two crossings below $x_\epsilon$.	
\end{proof}
\subsection{Numerical illustration}\label{subsec:num1}
How $\psi_\lambda(x)$ in~(\ref{e213}) varies with $x$ is shown in 
Figure \ref{fig:phi} 
under different combinations of risk parameters. Detailed conditions
and assumptions are recorded in Appendix \ref{appenB} along with the 
simulation algorithm used.
The portfolio is Poisson/Gamma with $50$ 
claims expected  annually and with individual losses 
on average $10$ with standard deviation $15$
which allows for huge losses. The model for $(X,Z)$ is based on the
Clayton Copula (consult Appendix \ref{appenB}) with a Gamma distribution for $Z$
with mean $1$ and standard deviation $0.3$. The market factor is
\begin{displaymath}
M(Z)=(1+\gamma^{\textrm{re}})\frac{e^{\omega Z}}{\E(e^{\omega Z})},
\end{displaymath}
with $\gamma^{\mbox{\tiny re}}$ and  $\omega$ parameters that are varied.
Note that $E\{M(Z)\}=1+\gamma^{\mbox{\tiny re}}$, and the B\"{u}hlman condition
corresponds to $\gamma^{\mbox{\tiny re}}=0$.

The default set of parameters in Figure \ref{fig:phi} is $(\beta, \theta, \gamma^{\textrm{re}}, \omega,\lambda)= (0.06 ,10, 0.2, 0.1, 0.1)$,
and the cost of capital $\beta$ is varied (top left), 
the price on risk $\lambda$ (top right), $\gamma^{\mbox{\tiny re}}$
(bottom left) whereas finally (bottom right) $\psi_\lambda(x)$ is shown 
for several parameter 
combinations. The optimum is in all but one case a one-layer
solution ending at the $1-\epsilon$ percentile $x_\epsilon$ 
($\psi_\lambda(x)>0$ from some lower limit up to $x_\epsilon$) 
or reinsurance everywhere $(\psi_\lambda(x)>0$
for $x<x_\epsilon$) or no reinsurance at all ($\psi_\lambda(x)<0$ everywhere).
The exception is the parameter combination with $\gamma^{\mbox{\tiny re}}=0$
so that $E\{M(z)\}=1$. Now $\psi_\lambda(0)>0$, 
and there is a $\bf b$-layer in the beginning and then an
$\bf a$-layer ending at $x_\epsilon$.

\begin{figure}[!t]	
	\centering
	\begin{minipage}{0.5\textwidth}
		\centering
		\includegraphics[width=\textwidth]{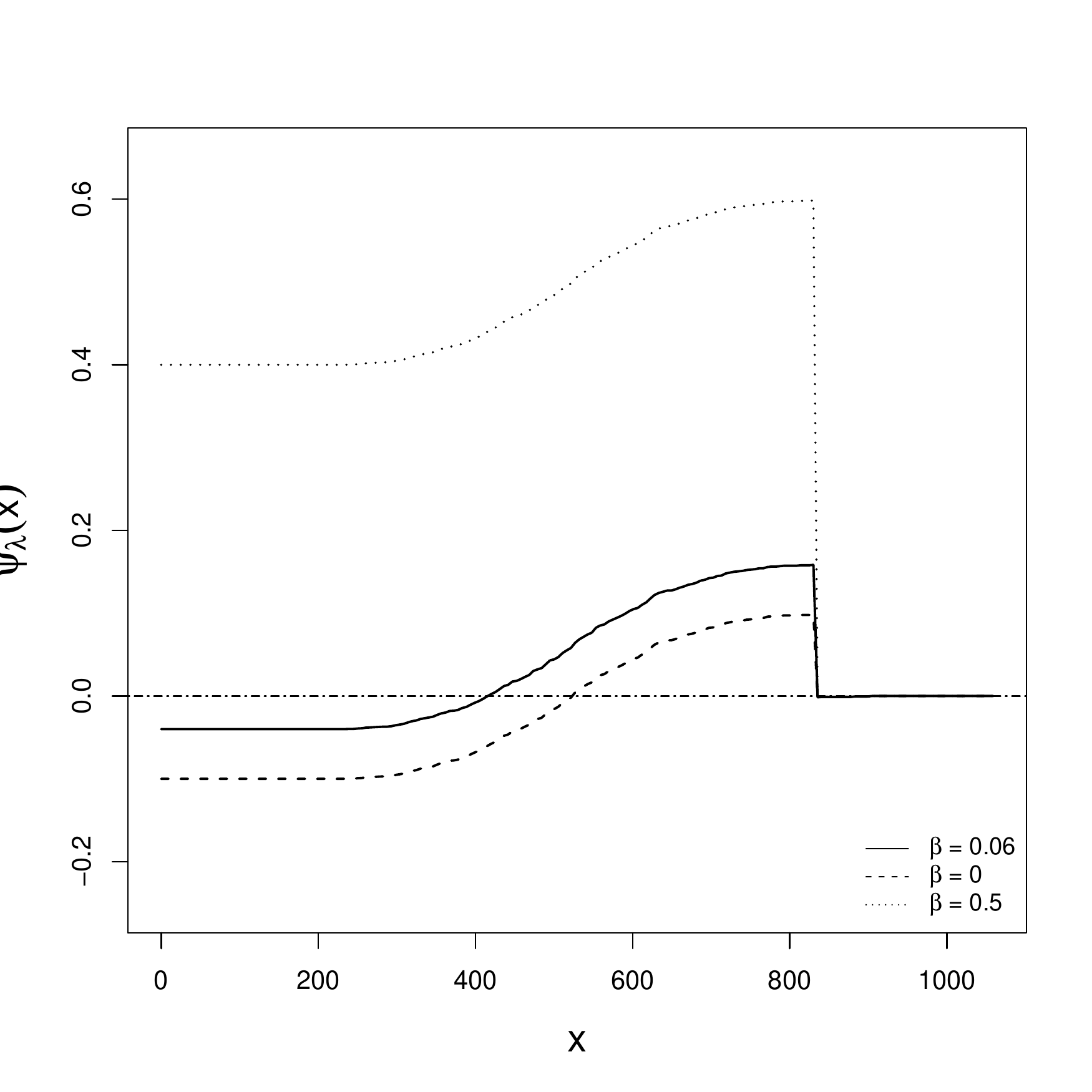}
	\end{minipage}%
	\begin{minipage}{0.5\textwidth}
		\centering
		\includegraphics[width=\textwidth]{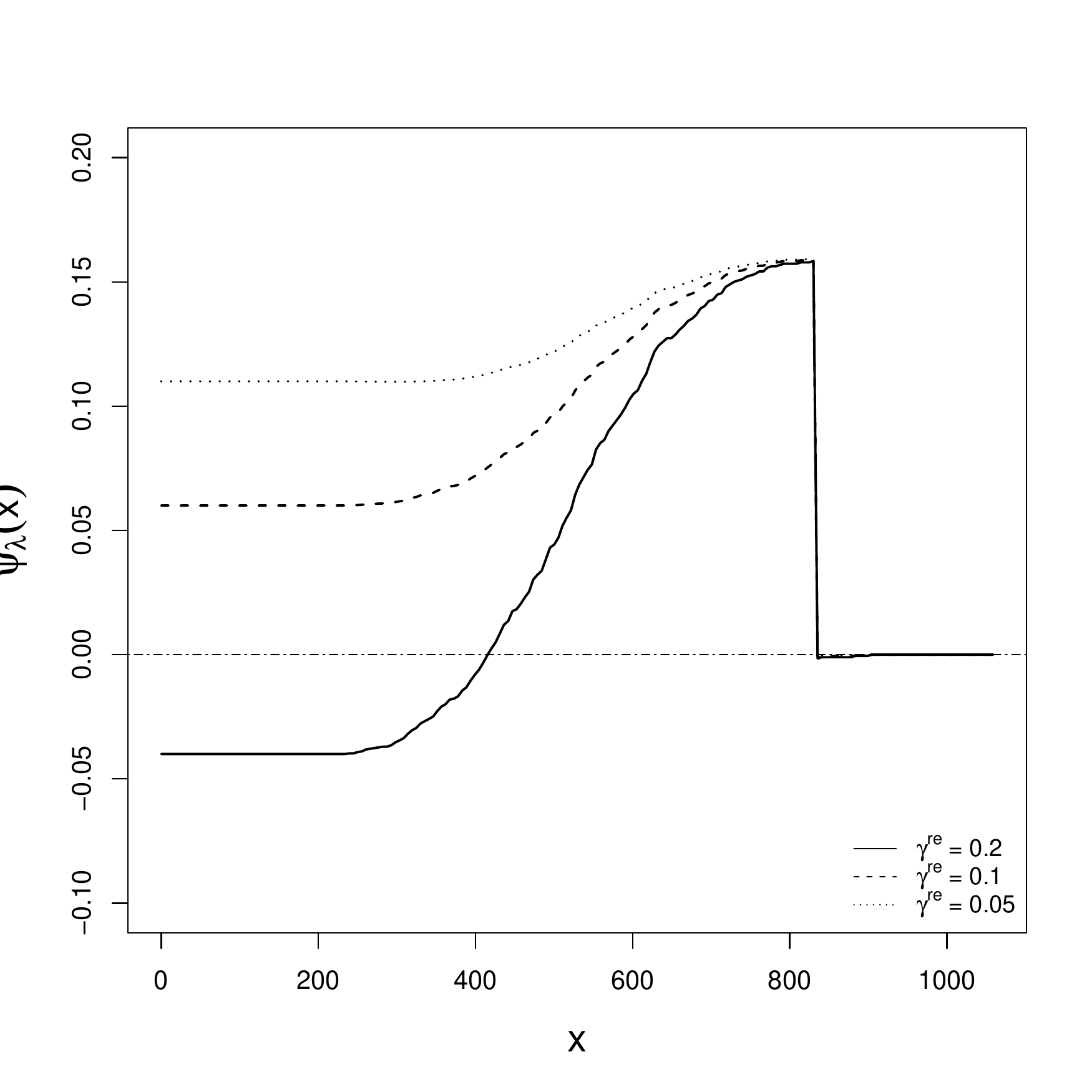}
	\end{minipage}
	\begin{minipage}{0.5\textwidth}
		\centering
		\includegraphics[width=\textwidth]{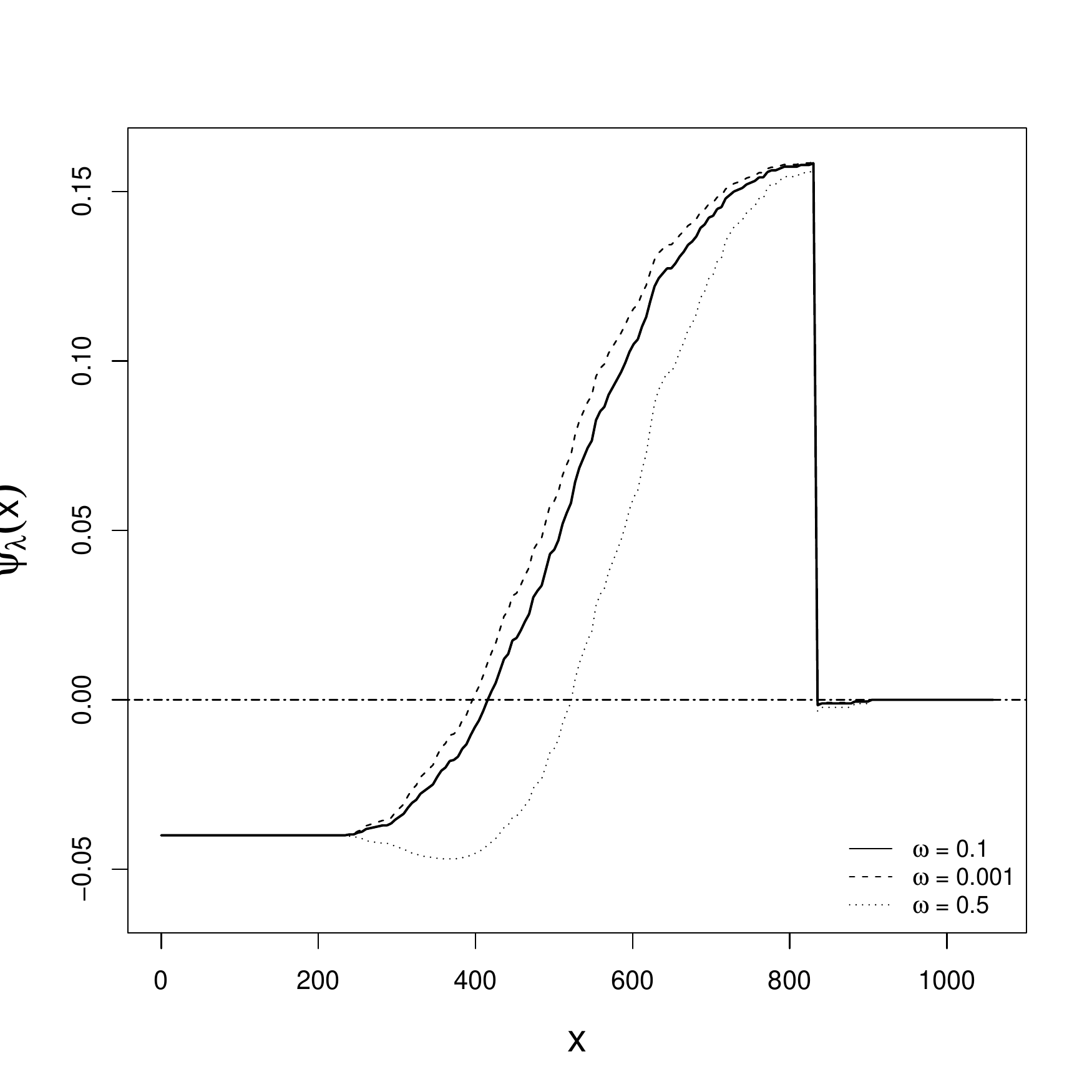}
	\end{minipage}%
	\begin{minipage}{0.5\textwidth}
		\centering
		\includegraphics[width=\textwidth]{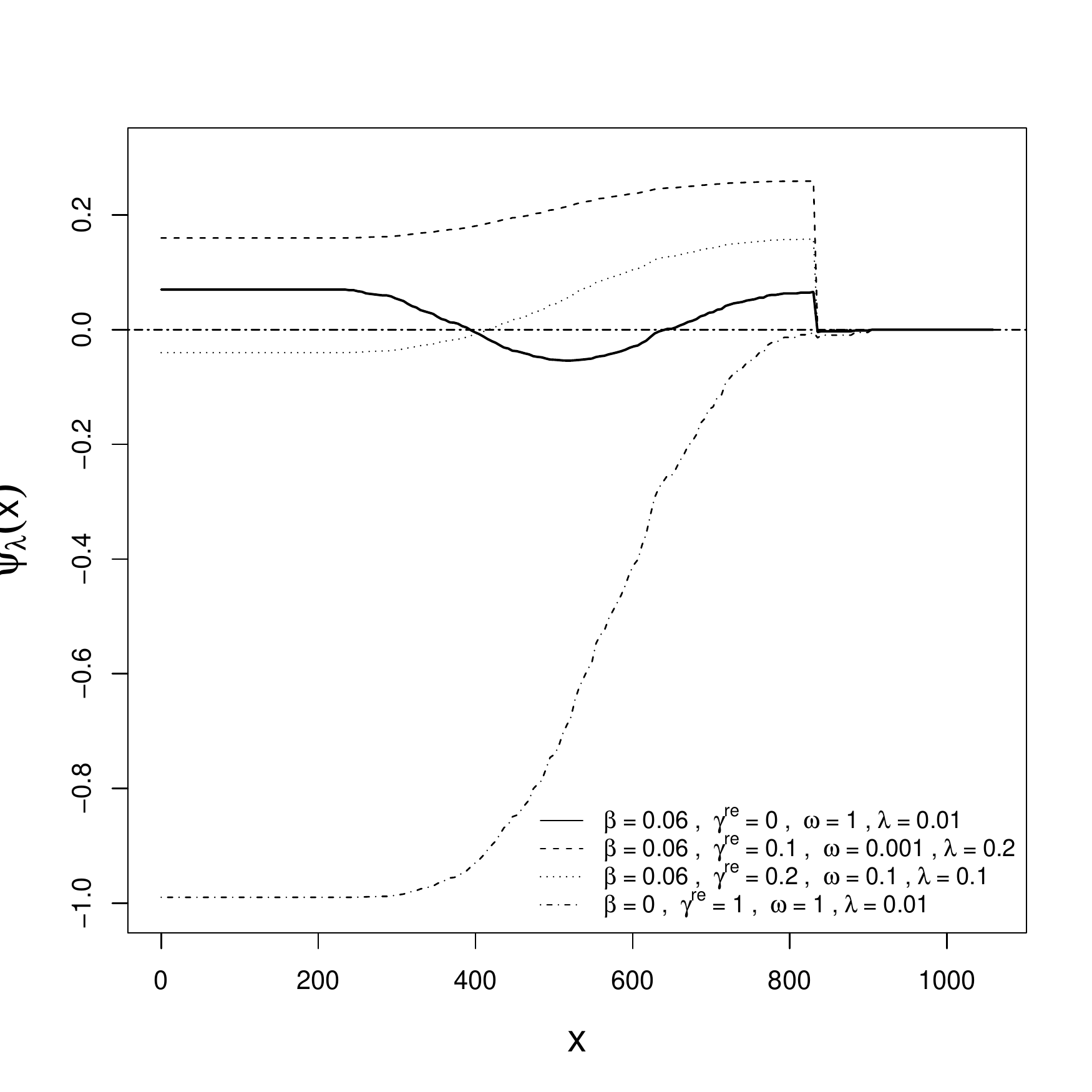}
	\end{minipage}
	\caption{Plot of $\psi_\lambda(x)$ against $x$ with $\beta$ (top left),
$\lambda$ (top right), $\gamma^{\mbox{\tiny re}}$ (bottom left) varied and 
with miscellaneous parameter combinations
(bottom right). The default set of parameters and conditions are in the text.}
	\label{fig:phi}
\end{figure}
\subsection{Risk over expected gain}\label{subsec:3.3}
Another criterion and the main focus in this paper is the ratio of
risk over expected surplus; i.e.
\begin{equation}
{\cal C}(I)=\frac{\rho(I)}{{\cal G}(I)}\hspace*{1cm}\mbox{conditioned on}\hspace*{1cm} {\cal G}(I) >0,
\label{e215}
\end{equation}
We search among insurance functions $I$ 
for which ${\cal G}(I)>0$, and ${\cal C}(I)$
is to be minimized among such contracts.
Inserting~(\ref{e23}),~(\ref{e26}) and~(\ref{e211}) yields the more
explicit form
\begin{equation}
{\cal C}(I)=\frac{x_\epsilon-I(x_\epsilon)}
{\gamma E(X)-\int_0^\infty  K\{F(x)\}dI(x)-\beta\{(x_\epsilon- I(x_\epsilon)\}}
\label{e216}
\end{equation}
where it is part of the optimization problem to keep the denominator positive.
Though not used much
in academic literature the criterion reflect  
industrial thinking very well as a tool to 
minimize risk per money unit expected gain. Its optimal solutions are 
still located
on an efficient frontier of the Markowitz type
and belong to the same class as those maximizing risk-adjusted surplus as the
following 
consequence of Proposition \ref{prop1} shows: 
\\
\begin{prop}\label{prop2}
Suppose Assumption (i) in Condition \ref{condi1} is true. Then there exists 
for any feasible reinsurance function $I(x)$
for which ${\cal G}(I)>0$ a one or two layer
function $I_{{\bf a b}}(x)=I_{\bf a}(x)+I_{\bf b}(x)$  so that
$b_1=0,$ $b_2\leq a_1,$ $a_2=x_\epsilon $, ${\cal G}(I_{\bf ab})={\cal G}(I)$ 
and
${\cal C}(I_{\bf ab})\leq {\cal C}(I)$.
\end{prop}
\begin{proof} 
Let
$\tilde{I}(x)=I(x)$ if $x\leq x_\epsilon$ and $=I(x_\epsilon)$ if $x>x_\epsilon$.
Value at Risk is then the same under both $I$ and $\tilde{I}$
whereas 
\begin{displaymath}
\int_0^\infty K\{F(x)\}dI(x)\geq\int_0^\infty K\{F(x)\}d\tilde{I}(x)
\end{displaymath}
since the contribution above $x_\epsilon$ is cut off on the right and 
$K(u)\geq 0$. It follows from~(\ref{e216}) that 
\begin{displaymath}
{\cal C}(\tilde{I})\leq
{\cal C}(I).
\end{displaymath}
The idea now is to construct 
a reinsurance function  $I_\lambda(x)$ satisfying~(\ref{e213})
for some $\lambda>0$
so that ${\cal G}(I_\lambda)={\cal G}(\tilde{I})$. 
Then by Corollary \ref{coro1} $I_\lambda=I_{\bf ab}$
for some pair of coefficients $\bf a$ and $\bf b$,
and since this contract maximizes
risk-adjusted surplus 
\begin{displaymath}
{\cal G}(I_{\bf ab})-\lambda \rho(I_{\bf ab})\geq {\cal G}(\tilde{I})-\lambda \rho(\tilde{I})
\qquad\mbox{with}
\qquad {\cal G}(I_{\bf ab})={\cal G}(\tilde{I}),
\end{displaymath}
denoting Value at Risk by $\rho(I_{\bf ab})$ and $\rho(\tilde{I})$.
Hence $\rho(I_{\bf ab})\leq \rho(\tilde{I})$ which implies
\begin{displaymath}
{\cal C}(I_{\bf ab})=\frac{\rho(I_{\bf ab})}{{\cal G}(I_{\bf ab})}\leq 
\frac{\rho(\tilde{I})}{{\cal G}(I)}={\cal C}(\tilde{I})\leq {\cal C}(I)
\end{displaymath}
which was to be proved. 
To construct $I_\lambda(x)$ 
let 
$\lambda_1=0$, and note that 
$K(u)\geq 0$ implies $\psi_{\lambda_1}(x)<0$ for all $x$
so that Proposition \ref{prop1} implies that the optimal 
reinsurance function when $\lambda_1=0$ is cost of risk   is
$I_1(x)=0$ everywhere. On the other hand if $\lambda_2$ is large enough,
$\psi_{\lambda_2}(x)>0$ for $x\leq x_\epsilon$ and the optimum contract now
is $I_2(x)=x$ for $x\leq x_\epsilon$ and $I_2(x)=x_\epsilon$
for $x> x_\epsilon$. The construction ensures that
\begin{displaymath}
0=\int_{0}^\infty K\{F(x)\}dI_1(x)\leq\int_{0}^{x_\epsilon} K\{F(x)\}d\tilde{I}(x)
\leq\int_{0}^{x_\epsilon} K\{F(x)\}dx=
\int_{0}^\infty K\{F(x)\}dI_2(x)
\end{displaymath}
since $K(u)\geq 0$.
It follows that
${\cal G}(I_1)\geq {\cal G}(\tilde{I})\geq {\cal G}(I_2)$. 
But if we allow $\lambda$ to grow from
$\lambda_1=0$ to $\lambda_2$ there will on continuity be a $\lambda$
in between so that 
${\cal G}(I_{\lambda})={\cal G}(\tilde{I})$
which 
completes the proof.
\end{proof}

\section{Large portfolio asymptotics}\label{sec:large_port}
\subsection{Introduction}\label{subsec:opt4}
The search for an optimum reinsurance function was above reduced  to the class
of two-layer ones $I_{\bf ab}$ with $b_1=0$ and $a_2=x_\epsilon$,
and the
aim now is to simplify further when portfolios are large.  Suppose 
$X_J=Y_1+\dots+Y_J$ with $Y_1,\dots,Y_J$ individual risks and 
$J$ large.
It might be possible
to cover situations where
$Y_1,\dots,Y_J$ are dependent through some common random  factor, but
that will not be done, and it is assumed that
$Y_1,\dots,Y_J$ are independent though not  
identically distributed.
The distribution
function $F(x)$ of $X_J$  then becomes Gaussian as 
$J\rightarrow \infty$ by the central limit theorem and its Lindeberg
extension which is almost always satisfied. We shall from now on
write $F_J(x)=F(x)$ to emphasize the importance of 
$J$, similarly $x_{\epsilon J}=x_\epsilon$ for the $1-\epsilon$
percentile,
${\cal G}_{J}(I)={\cal G}(I)$, $\rho_J(I)=\rho(I)$ 
and ${\cal C}_{J}(I)={\cal C}(I)$
and even $a_{1J}=a_1$ and $a_{2J}=a_2$ for the coefficients. 
\\\\
The detailed
mathematical calculations are relegated to Appendix \ref{appendix}, but the 
crux of the approach is 
the centered and normalized variable
\begin{equation}
X_J^0=\frac{X_J-J\xi}{\sqrt{J}\sigma}
\label{e42}
\end{equation}
with $\xi$ average mean and $\sigma^2$
average variance
of the individual risks $Y_1,\dots,Y_J$. 
For the ensuing argument
it doesn't matter that
in reality $\xi=\xi_J$ and $\sigma=\sigma_J$ depend on $J$
as long as they converge
to fixed values $\xi$ and $\sigma$ as $J\rightarrow\infty$. 
This minor
complication is ignored.
The
distribution function  $F_J^0(x)$ of $X_J^0$ 
with percentiles $x_{\epsilon J}^0$
starts at
$x_{J}^0=-\sqrt{J}\xi/\sigma$ (zero below), 
and we have   the elementary relationships
\begin{equation}
F_J(x)=F_{J}^0\left(\frac{x-J\xi}{\sqrt{J}\sigma}\right)
\hspace*{1cm}\mbox{and}\hspace*{1cm}
x_{\epsilon J}^0=\frac{x_{\epsilon  J}-J\xi}{\sqrt{J}\sigma}.
\label{e43}
\end{equation}
The reason for introducing $X_J^0$ is that
$F_J^0(x)\rightarrow \Phi(x)$ and 
$x_{\epsilon  J}^0\rightarrow \phi_\epsilon$
as $J\rightarrow \infty$ where
$\Phi(x)$ and $\phi_\epsilon$ are distribution function and percentile for the 
standard normal. It is convenient to work  with similar versions
for the coefficients; i.e.
\begin{equation}
a_{1J}^0=\frac{a_{1J}-J\xi}{\sqrt{J}\sigma}
\hspace*{1cm}\mbox{and}\hspace*{1cm}
a_{2J}^0=\frac{a_{2J}-J\xi}{\sqrt{J}\sigma}
\label{e44}
\end{equation}
and $a_{2J}^0=x_{\epsilon  J}^0$ when $a_{2J}=x_{\epsilon J}$ is the 
optimal upper cut-off point. Similar normalized coefficients
$b_{1J}^0$ and $b_{2J}^0$ are introduced from $b_{1J}$ and $b_{2J}$ below.
\subsection{A key condition}\label{subsec:4.2}
Consider the one-layer contract $I_{{\bf a}_J}$
with limits $a_{1J}$ and $a_{2J}=x_{\epsilon J}$.
The expected reinsurer surplus~(\ref{e211}) is
\begin{displaymath}
\pi(I_{{\bf a}_J})-E_J\{I_{{\bf a}_J}(X_J)\}
=\int_{0}^\infty K\{F_J(x)\}dI_{{\bf a}_J}(x)
=\int_{a_{1J}}^{x_{\epsilon J}} K\{F_J(x)\}dx
\end{displaymath}
or after changing the integration variable to
$t=(x-J\xi)/(\sqrt{J}\sigma)$ in the last integral
\begin{equation}
\pi(I_{{\bf a}_J})-E_J\{I_{{\bf a}_J}(X_J)\}
=\sqrt{J}\sigma\int_{a_{1J}^0}^{x^0_{\epsilon J}} K\{F_J^0(t)\}dt.
\label{e45}
\end{equation}
Note that $I_{{\bf a}_J}(x_{\epsilon _J})=(x_{\epsilon J}-a_{1J})_+$
so that Value at Risk is 
$x_{\epsilon J}-I_{{\bf a}_J}(x_{\epsilon _J})=\min(a_{1J},x_{\epsilon J})$
which
after inserting for $a_{1J}$ and
$x_{\epsilon _J}$
becomes
\begin{equation}
x_{\epsilon J}-I_{{\bf a}_J}(x_{\epsilon J})=J\xi+\sqrt{J}\sigma
\min(a_{1J}^0,x_{\epsilon J}^0).
\label{e46}
\end{equation}
By~(\ref{e26})
\begin{displaymath}
{\cal G}_J(I_{{\bf a}_J})=\gamma E(X_J)-
\{\pi(I_{{\bf a}_J})-E_J\{I_{{\bf a}_J}(X_J)\}\}
-\beta \{x_{\epsilon J}-I_{{\bf a}_J}(x_{\epsilon J})\},
\end{displaymath}
and after inserting $E(X_J)=J\xi$,~(\ref{e45}) and~(\ref{e46}).
\begin{equation}
{\cal G}_J(I_{{\bf a}_J})=J\xi(\gamma -\beta)-
\sqrt{J}\sigma\left(\int_{a_{1J}^0}^{x^0_{\epsilon J}} K\{F_J^0(t)\}dt+
\beta
\min(a_{1J}^0,x_{\epsilon J}^0)
\right)
\label{e47}
\end{equation}
which will be later needed to prove Proposition \ref{prop4.1} below.
\\\\
Suppose $a_{1J}=0$ with $a_{2J}=x_{\epsilon J}$ still.
The reinsurer then takes all risk up to the $1-\epsilon$
percentile so that Value at Risk for
the insurer is zero. 
In~(\ref{e47}) $a_{1J}^0=x_{J}^0=-\sqrt{J}\xi/\sigma$,
and when this is inserted, $\beta$ vanishes (as it should) and
\begin{displaymath}
{\cal G}_J(I_{{\bf a}_J})=J\xi\gamma-
\sqrt{J}\sigma\int_{x_{J}^0}^{b^0_{2 J}} K\{F_J^0(t)\}dt
\end{displaymath}
or
\begin{equation}
\frac{{\cal G}_J(I_{{\bf a}_J})}{J\xi}=\gamma+\frac{1}{x_J^0}
\int_{x_{J}^0}^{x_{\epsilon J}^0} K\{F_J^0(t)\}dt\rightarrow \gamma-K(0)
\quad\mbox{as}\quad J\rightarrow\infty.
\label{e46a}
\end{equation}
The limit follows from Lemma A.2 in  Appendix A and 
is also a simple consequence of l'H$\hat{\mbox{o}}$pital's rule if
$K\{F_J^0(x_J^0)\}=K(0)>0$ so that 
the integral in~(\ref{e46a}) becomes infinte as $J\rightarrow \infty$.
\\\\
Surely this 
suggests $K(0)\geq \gamma$?
Otherwise the insurer by expanding the portfolio and reinsuring everything up
to Value at Risk can earn money without having to put up any solvency 
capital at all. This isn't quite arbitrage since risk 
above Value at Risk still rests with the insurer, but it is   
for large portfolios
rather close to it. It seems
unlikely that the market should allow reinsurance risk to be priced
so cheaply that $K(0)<\gamma$, and
$K(0)\geq \gamma$  becomes a fair assumption.
The formal condition in~(\ref{e211a}) excludes the 
case $K(0)=\gamma$.
Now Value at Risk is $0$ for large portfolios 
when the insurer 
reinsures everything below
so that its
ratio over expected surplus
is $0$ too (and hence minimized). For large portfolios 
the situation has  
become trivial and uninteresting and need not be considered.
\subsection{Optima for large portfolios}\label{subsec:4.3}
It was shown above that $x_{\epsilon J}$
is the upper cut-off point for the optimal 
reinsurance function, and it will now turn out that for 
large portfolios the $1-\delta$
percentile
$x_{\delta J}$ is the lower one where $\delta$ was 
defined in Lemma \ref{lem2.1};
consult, in particular~(\ref{e211b}) from which it follows that
$\delta$ depends on the reinsurance pricing regime, but on not the actual
distribution of risks. Define
\begin{equation}
\hat{a}_{1J}=\min(x_{\delta J},x_{\epsilon J})
\hspace*{1cm}\mbox{and}\hspace*{1cm}\hat{a}_{2J}=x_{\epsilon J},
\label{e41}
\end{equation}
and let $\hat{\bf a}_J=(\hat{a}_{1J},\hat{a}_{2J})$.
Throughout this section 
$\hat{\bf a}_J$ will represent different approximations to the
true optimum.
For the corresponding one-layer
reinsurance function $I_{\hat{\bf a}_J}$ under~(\ref{e41})
there is the following result:
\\
\begin{prop}\label{prop4.1}
	If Condition \ref{condi1} is true and 
$\gamma>\beta$, the reinsurance function~(\ref{e41}) 
satisfies the following:
For any $\eta>0$  there exists
$J_\eta$  so that if
$J>J_\eta$, then ${\cal G}_J(I_{\hat{\bf a}_J})>0$  and
${\cal C}_J(I_{\hat{\bf a}_J})\leq
{\cal C}_{J}(I)+\eta$ for any reinsurance function $I$ for which
${\cal G}_{J}(I)>0$.
\end{prop}

In one word ${\cal C}_J(I_{\hat{\bf a}_J})$ can't in the limit exceed 
${\cal C}_{J}(I)$
for any other feasible reinsurance function $I$. 
The result provides a simple recipe for optimum reinsurance  
when portfolios are large with the second layer $I_{{\bf b}_J}$ not included at all.
Note that it could happen that $\delta<\epsilon$ so that
$x_{\delta J}>x_{\epsilon_J}$ in~(\ref{e41}) which implies 
$\hat{a}_{1J}=\hat{a}_{2J}$, 
and the optimum for the insurer is now to carry all risk. This happens when 
reinsurance is expensive.
The extra assumption $\gamma>\beta$ is always satisfied in practice; surely
no insurer would operate if the loading $\gamma$ 
in the primary insurance market didn't exceed the cost $\beta$ 
per money unit of keeping
solvency capital. The proposition is proved in Appendix \ref{appendix}.

\subsection{Degradation asymptotics}\label{subsec:4.4}
The one-layer reinsurance function $I_{\hat{\bf a}_J}$ 
with the two percentiles in~(\ref{e41}) as limits is optimal 
as $J\rightarrow \infty$,
but how much is the solution degraded for finite $J$? 
It may be measured against the true optimum 
$I_{\tilde{\bf a}_J}$ based on the coefficients $\tilde{\bf a}_J$
that minimizes ${\cal C}_J(I_{{\bf a}_J})$. Note that this is the relevant
comparison since there is by Proposition \ref{prop4.1} no second layer 
when $J$ is large enough.
The degree of
degradation in $I_{\hat{\bf a}_J}$ is therefore 
\begin{equation}
{\cal D}_J(I_{\hat{\bf a}_J})={\cal C}_J(I_{\hat{\bf a}_J})-
{\cal C}_J(I_{\tilde{\bf a}_J}),
\label{e415}
\end{equation}
which is non-negative and 
$\rightarrow 0$ by Proposition 4.1
as $J\rightarrow \infty$. The following result provides 
the rate:
\\
\begin{prop}\label{prop4.2}
	If Condition \ref{condi1} is true and $\gamma>\beta $, 
the degradation~(\ref{e415}) is under~(\ref{e41})
	\begin{equation}
	{\cal D}_J(I_{{\hat{\bf{a}}_J}})=\frac{\zeta_1}{J^{3/2}}+o(1/J^{3/2})
	\label{e416}
	\end{equation}
	where
	\begin{equation}
	\zeta_1=-\frac{1}{2}\frac{B_1^2/\{\xi^3(\gamma-\beta)^2\}}
	{K'(1-\delta)\Phi'(\phi_{\delta })},\hspace*{1cm}
	B_1=\sigma^2\left(\int_{\phi_\delta}^{\phi_\epsilon}K\{\Phi(y)\}dy+\phi_{\delta } K(1-\delta)\right).
	\label{e417}
	\end{equation}
\end{prop}
Note that $\zeta_1> 0$ since $K'(1-\delta)<0$ by Lemma \ref{lem2.1}, and the 
first term on the right in~(\ref{e416}) is positive as it should.
Consult Appendix \ref{appenA2} for the proof. 
\\\\
Simple, accurate calculation
of the coefficients in~(\ref{e41}) is available through Monte Carlo.
If $X_{J(1)}^*\leq\dots\leq X^*_{J(m)}$ is an ordered sample of simulations of
$X_J$ take
$\hat{a}_{1J}=\min(X^*_{J(m_{\delta})},X^*_{J(m_\epsilon)})$ and
$\hat{a}_{2J}=X^*_{J(m_\epsilon)}$ where $m_\delta=(1-\delta)m$ and
$m_\epsilon=(1-\epsilon)m$. But
what happens if Gaussian percentiles are used instead 
so that no Monte Carlo is needed at all? Now
instead of~(\ref{e41})
\begin{equation}
\hat{a}_{1J}=J\xi+\sqrt{J}\sigma\min(\phi_\delta,\phi_\epsilon)
\hspace*{1cm}\mbox{and}\hspace*{1cm}
\hat{a}_{2J}=J\xi+\sqrt{J}\sigma\phi_\epsilon.
\label{e418}
\end{equation}
Had the risk $X_J$ been strictly Gaussian,
Proposition \ref{prop4.2}  would still apply, but in practice
this is only an approximation, and we must suspect a lower degradation rate.
The following proposition is proved in Appendix \ref{appenA3}:
\\
\begin{prop}\label{prop4.3}
Suppose Condition \ref{condi1} is true and that $\gamma>\beta$. If
        $x_{\epsilon J}\geq \phi_\epsilon$
	 for large $J$,
	then the  degradation under~(\ref{e418}) is
	\begin{equation}
	{\cal D}_J^N(I_{{\hat{\bf{a}}_J}})=\frac{1}{\sqrt{J}}\zeta_2
	(x_{\epsilon J}^0-\phi_\epsilon)+o(1/J)
	\qquad\mbox{where}\qquad
	\zeta_2=\frac{\sigma\{\gamma-K(1-\epsilon)\}}{\xi(\gamma-\beta)^2}.
	\label{e421}
	\end{equation}
\end{prop}
Note that Lemma \ref{lem2.1} established  that $K(1-\epsilon)<\gamma$
so $\zeta_2>0$.

The main contribution to the
degradation is thus caused by the  discrepancy $x_{\epsilon J}^0-\phi_\epsilon$
at the upper percentile which may be approximated by the 
Cornish-Fischer correction 
\begin{equation}
x_{\epsilon J}=\phi_\epsilon+\frac{1}{\sqrt{J}}p(\phi_\epsilon)+
o(1/\sqrt{J})
\quad\mbox{where}\quad p(x)=\kappa(x^2-1)/6;
\label{e419}
\end{equation}
consult (for example) Section 2.5 in \cite{hall1992}.
The coefficient  $\kappa$ is the average skewness 
of the individual risk variables underlying the portfolio sum $X_J$, and
the usual situation is $\kappa>0$; consult Chapter 10 in
\cite{bolviken_2014}, for an expression for $\kappa$. 
Then $x_{\epsilon J}^0>\phi_\epsilon$ 
as assumed in Proposition \ref{prop4.3},
and
the  degradation now becomes
\begin{equation}
{\cal D}_J^N(I_{\hat{\bf a}_J})=\frac{\zeta_2p(\phi_\epsilon)}{J}+o(1/J).
\label{e421}
\end{equation}
Asymptotic results can also be derived when 
$\kappa\leq 0$ which is so rare that it has little practical interest.
\\\\
Proposition \ref{prop4.3} indicates that the accuracy is enhanced when
a better approximation of $x_{\epsilon J}^0$ than $\phi_\epsilon$ 
is used.
Suppose
in a manner 
resembling the Normal
Power method of reserving in property insurance~(\ref{e418})
right is replaced by
\begin{equation}
\hat{a}_{1J}=J\xi+
\sqrt{J}\sigma\min(\phi_\delta,\phi_\epsilon)+\sigma 
p\{\min(\phi_\delta\,\phi_\epsilon)\}
\quad\mbox{and}\quad
\hat{a}_{2J}=J\xi+\sqrt{J}\sigma\phi_\epsilon+\sigma p(\phi_\epsilon)
\label{e422}
\end{equation}
with the Cornish-Fisher correction term added. The error in the approximation
of  $x_{\epsilon J}^0$ is then of order $o(1/J)$, and it follows from
Proposition \ref{prop4.3} that
the degradation ${\cal D}^{NP}(I_{\hat{\bf a}_J})$ now is of order $1/J^{3/2}$.
\\\\
These results also tell something about the impact of
model error.
The Poisson distribution, supported by the Poisson point process,
is often a reasonable choice for claim numbers, but there is
rarely much theory
behind the choice of a typical two-parameter family for claim size.
Suppose two such families are calibrated so that mean and standard 
deviation match. 
The same Gaussian distribution appears in the limit as
$J\rightarrow \infty$ in either case, and the discrepancies
in the optimum value of the criterion are thus of order
$O(1/J)$ and not very large for $J$ of some size.

\section{Numerical study}\label{sec:numerical}
\subsection{Example and conditions}\label{subsec:5.1}
The Monte Carlo study presented is this section is based on the market factor
$M(Z)=1+\gamma^{\mbox{\tiny re}}$ independent of $Z$ 
so that in~(\ref{e29}) right 
$W(u)=1+\gamma^{\mbox{\tiny re}}$ which yields
 in~(\ref{e210})
\begin{displaymath}
K(u)=\gamma^{\mbox{\tiny re}}(1-u),
\end{displaymath}
and the $\delta$-percentile
in Lemma \ref{lem2.1}
which is the solution of $K(1-\delta)=\gamma$, becomes
$\delta=\gamma^{\mathrm{re}}/\gamma$. 
Numerical values were $\gamma=0.1$ and
$\gamma^{\mathrm{re}}=0.2$ so that $\delta=0.5$ which means that the 
large-portfolio approximations of the optimal reinsurance function
use the $50\%$ percentile of $X$ as lower limit.
The other percentile was $1-\epsilon=99\%$. Cost of capital was taken as
$\beta=0$.
\\\\
The claim number was Poisson distributed with claim frequency 
per policy $\mu=0.05$, and the portfolio
size varied between
$J=10^3,10^4$ and $10^5$ policies 
representing small, medium and large portfolios
corresponding to $J\mu=50,500$ and $5000$ expected incidents. As model for the
individual losses we have taken three classic
distributions with strong skewness to the right; i.e
Gamma, log-normal and Pareto. 
The probability density functions for Gamma and Pareto were
respectively
\begin{displaymath}
g(y)= \frac{y^{\alpha-1}e^{-y\alpha /\xi}}{(\xi\alpha)^{\alpha}\Gamma(\alpha)}
\quad \mbox{and}
\quad
g(y) = \frac{\alpha/\{\xi(\alpha-1)\}}{(1+y/\{\xi(\alpha-1)\})^{\alpha+1}}
\end{displaymath}
for $y>0$ whereas for the log-normal 
$\log(Y)$ was normal with mean $\alpha$
and variance $2(\log(\xi)-\alpha)$. This way of
parameterizing means that $\xi$ is mean loss per event
in all three cases whereas 
$\alpha$ determines variation. The models were 
calibrated so that $\xi=10$ and sd$(Y)=15$ which mean that
$\alpha=0.44$ (Gamma), $\alpha=1.71$ (log-normal)
and
$\alpha=3.60$ (Pareto) with strong skewness in all three cases,
respectively $3.00$ (Gamma), $7.88$ (log-normal) and $5.78$ (Pareto).
The extreme right tail is heaviest for the Pareto distribution despite 
its skewness being lower than for the log-normal. 

\subsection{Results}
The optimum Value at Risk over expected surplus had to be optimized
numerically as a benchmark against which the approximations could be
evaluated. Recall that 
the upper limit should be the $1-\epsilon$
percentile so the optimization was a simple one-dimensional one
to find the lower limit. 
Monte Carlo was needed to compute the criterion.
The number of simulations was
$m=10^6$, more than enough to keep
Monte Carlo error at a comfortably low level. 

\begin{table}
	\centering
	\small
	\begin{tabular}{lccccc}
		\hline
		\hline
		Model & $J\mu$&${\cal C}_J(I_{\tilde{{\bf a}}_J})$ &${\cal D}_J(I_{{\hat{\bf{a}}_J}})$&${\cal D}^N_J(I_{{\hat{\bf{a}}_J}})$&${\cal D}^{NP}_J(I_{{\hat{\bf{a}}_J}})$\\
		\hline
		\multirow{3}{*}{Gamma}&$5$&21.52&$1.53\times 10^{1}$&$2.02\times 10^{1}$&$1.63\times 10^{1}$\\
		&$50$&12.46&$1.00\times 10^{-1}$&$9.82\times 10^{-1}$&$1.00\times 10^{-1}$\\
		&$500$&10.68&$2.33\times 10^{-3}$&$8.54\times 10^{-2}$&$2.24\times 10^{-3}$\\
		&$5000$&10.21&$6.84\times 10^{-5}$&$8.48\times 10^{-3}$&$5.58\times 10^{-5}$\\
		\hline
		\multirow{3}{*}{Lognormal}&$5$&20.33&$8.65\times 10^{0}$&$2.00\times 10^{1}$&$3.25\times 10^{1}$\\
		&$50$& 12.39&$9.09\times10^{-2}$&$1.64\times 10^{0}$ &$1.18\times10^{-2}$\\
		&$500$&10.68&$2.29\times 10^{-3}$&$1.67\times 10^{-1}$&$3.06\times 10^{-3}$\\
		&$5000$&10.21&$6.66\times 10^{-5}$&$1.57\times 10^{-2}$&$1.08\times 10^{-4}$\\
		\hline
		\multirow{3}{*}{Pareto}&$5$&20.30&$8.77\times 10^{0}$&$1.91\times 10^{1}$&$2.17\times 10^{1}$\\
		&$50$&12.37&$9.02\times 10^{-2}$&$1.56\times 10^{0}$&$9.90\times 10^{-2}$\\
		&$500$&10.67&$2.27\times 10^{-3}$&$1.64\times 10^{-1}$&$2.09\times 10^{-3}$\\
		&$5000$&10.21&$6.62\times10^{-5}$&$1.75\times10^{-2}$&$5.11\times10^{-5}$\\
		\hline\hline
	\end{tabular}
	\caption{The degradation based on the three
approximations of the optimal $a_{1J}$ and $a_{2J}$ with different loss 
distributions and conditions as in Section \ref{subsec:5.1}.}\label{tab:5.1}
\end{table}
Main results are summarized
in Table \ref{tab:5.1} for different values of the expected number of incidents $J\mu$
and the three different loss distributions.
All the three 
approximations~(\ref{e41}),~(\ref{e418}) and~(\ref{e422}) have been 
evaluated and are
recorded as ${\cal D}_J(I_{{\hat{\bf{a}}_J}})$, ${\cal D}^N_J(I_{{\hat{\bf{a}}_J}})$,
${\cal D}_J^{NP}(I_{{\hat{\bf{a}}_J}})$, and these 
values in Columns $4-6$ must be judged against
the optimum of the Value at Risk over surplus ratio in Column $3$.
What counts is the ratios.
First note that the criterion itself
is strongly dependent on portfolio size with much higher risk over surplus 
when the expected number of incidents are small. The approximations
when $J\mu=5$ are useless, but that changes for larger portfolios
with the loss in Column $4$ and $6$ around 
$8\%$ when $J\mu=50$, $0.2\%$ when $J\mu=500$
and perhaps $0.0006\%$ when $J\mu=5000$. The normal approximation
in Column 5 is inferior to the two others as the results in 
Section \ref{subsec:4.3}
suggested.
Decay rates as $J$ grows
match the theoretical ones and are 
around $1/J$
for the Gaussian approximation in Column $5$ 
and  around $1/J^{3/2}$ for the two others with the latter
remarkably similar. Discrepancies between the three loss distributions
are minor. Since they were calibrated so that mean 
and standard deviation are equal, the experiment
testifies to the lack of importance of the shape of the distributions 
beyond the first two moments.

\section{Concluding remarks}
A large-portfolio approach has been introduced which leads 
to
a modification for the
market factor $M(Z)$ in
the B\"{u}hlman
pricing regime for reinsurance.
Instead of imposing the usual $E\{M(Z)\}=1$ we have assumed
$E\{M(Z)\}>  1+\gamma$ where $\gamma$ is the loading
in the primary market of the insurer. If this condition fails to hold, 
insurers can in large portfolios 
earn money with  no net  
solvency capital being needed, arguably  
an unlikely state of affairs. 
It was this  new condition that 
reduced the
optimum contracts  for the Value at Risk adjusted surplus  
in \cite{chi2017optimal}
to one or two-layer ones, and that 
applied to the Value at Risk over expected surplus
ratio as well. There was only one layer when the price on risk
$\lambda$ is below a threshold. If prices in the reinsurance
market is of the expected premium principle type with loading 
$\gamma^{\mbox{\tiny re}}$, the condition boils down to
$\lambda<\gamma^{\mbox{\tiny re}}-\beta$ with $\beta$ the cost of 
solvency capital. Our judgment is that this condition
might often be satisfied, but against that view
there is the fact that the
prices in the reinsurance market
are distinctly volatile and nor do they have so simple a structure
as a fixed loading.
\\\\
It has for large portfolios been shown that 
one-layer contracts are close to optimum in any case, and 
that the world of optimum reinsurance is under these circumstances 
an orderly one.
The end points of the best layer is now defined as fixed percentiles 
of the underlying risk variable
with the lower one determined by reinsurance prices. How far this
solution is from the true optimum was investigated theoretically
through large-portfolio studies that lead to degradation rates of order
$O(1/J^{3/2})$ when Monte Carlo approximations of the
exact percentiles are used and
$O(1/J)$ for Gaussian ones with $J$ the number of policies. 
There is even a Normal Power modification of the latter that achieves
$O(1/J^{3/2})$ too.
These results
were supported by numerical experiments in Section \ref{sec:numerical}
which suggested considerable robustness
with respect to the shape of the underlying 
claim severity distribution. The important thing for optimal 
reinsurance seems to be to
get mean and variance right.
\\\\
The studies in this paper can be extended along two lines. We conjecture
that similar results are obtained when Value at Risk is replaced by Conditional
Value at Risk. The main difference will be that the fixed percentile
$1-\epsilon$ for the upper limit will be replaced by larger one. Then there
is the condition of independent risks. They are in  many situations
some common random factor influencing all of them, for example 
a random claim frequency. 
Now the central limit 
theorem on which the present paper is based no longer holds. Portfolio losses
still  have a limit
distribution, but it is very different from the one in Section \ref{sec:large_port}. 
It would be of
practical interest to develop theory in this situation and investigate
how optimal reinsurance is influenced.

\bibliography{optre4_new.bib}
\appendix
\section{Proofs of asymptotics}\label{appendix}
\subsection{Proposition \ref{prop4.1}}
The following inequality is needed:
\\
\begin{lemma}\label{lemmaA1}
If $K(u)$  and $W(u)$  are as in Section \ref{sec:basics},
	and  $F(x)$ is a distribution function, then for all $a$, 
$b$ and $x$
\begin{equation}
|\int_{a}^{x} K\{F(t)\}dt+aK\{F(b)\}|\leq |x|\left(\int_0^1 |W(u)-1|du+
K\{F(b)\}\right).
\label{a1}
\end{equation}	
\end{lemma}

\begin{proof}
	Note that
	\begin{displaymath}
	\int_{a}^{x} K\{F(t)\}dt+aK\{F(b)\}=
	\int_{a}^x \{K\{F(t)\}-K\{F(b)\}\}dt+xK\{F(b)\}
	\end{displaymath}
	where the integral when inserting~(\ref{e210}) for $K(u)$ is
	\begin{displaymath}
	\int_{a}^{x}\int_{F(b)}^{F(t)}(W(u)-1)dudt
	=\int_{F(b)}^{F(x)}(W(u)-1)\int_{\max(F^{-1}(u),a)}^{x}dtdu
	\end{displaymath}
	\begin{displaymath}
	\hspace*{4.3cm}=\int_{F(b)}^{F(x)}(W(u)-1)\{x-\max(F^{-1}(u),a)\}du
	\end{displaymath}
	which means that the absolute value is bounded by $|x|\int_{0}^1|W(u)-1|du$,
	and~(\ref{a1}) follows.
\end{proof}
\begin{lemma}\label{lemmaA2}
	If $x_J^0=-\sqrt{J}\xi/\sigma$ and
$F_J^0(x)$ is the distribution function of the normalized risk variable
$X_J^0=(X-J\xi)/(\sqrt{J}\sigma)$, then
\begin{equation}
\frac{1}{x_J^0}\int_{x_J^0}^{x_{\epsilon J}^0} K\{F_J^0(t)\}dt+K(0)\rightarrow 0
\quad\mbox{as}\quad J\rightarrow \infty.
\label{a2}
\end{equation}
\end{lemma}
\begin{proof}
	Insert 
	$a=x_J^0=-\sqrt{J}\xi/\sigma$,
	$x=x_{\epsilon J}^0$, $b=x_J^0$ in~(\ref{a1}) and
	replace $F(x)$ with $F_J^0(x)$. Then
	\begin{displaymath}
	|\int_{x_J^0}^{x_{\epsilon J}^0} 
	K\{F_J^0(t)\}dt+x_{J}^0K(0)|\leq |x_{\epsilon J}^0|\left(\int_0^1 |W(u)-1|du+
	K(0)\right),
	\end{displaymath}
	and since $x_{\epsilon J}^0$ tends to the Gaussian percentile as 
	$J\rightarrow \infty$
	and $|W(u)-1|$ has finite integral, this implies~(\ref{a2}).
\end{proof}
The next lemma shows that one-layer contracts 
are better than two-layer ones when
portfolios are large:
\\
\begin{lemma}\label{lemmaA3}
	There exits for any $\eta>0$  some $J_\eta$
	so that if Condition \ref{condi1} is true,
	\begin{equation}
	\sup\{{\cal C}_J(I_{{\bf a}_J})-{\cal C}_J(I_{{\bf a}_J{\bf b}_J})\}<\eta
	\quad\mbox{when}\quad J>J_\eta
	\label{a3}
	\end{equation}
    where the sup is over all sequences of coefficients ${\bf a}_J$  and
	${\bf b}_J$  so that
	$0=b_{1J}\leq b_{2J}\leq a_{1J}\leq a_{2J}=x_{\epsilon J}$.
\end{lemma}
\begin{proof}
	Consider the reinsurance function
	$I_{{\bf a}_J{\bf b}_J}$ with coefficients as in the lemma
	for which 
	$I_{{\bf a}_J{\bf b}_J}(x_{\epsilon J})=x_{\epsilon J}-a_{1J}+b_{2J}$. 
	Hence Value at Risk becomes
	\begin{displaymath}
	x_{\epsilon J}-I_{{\bf a}_J{\bf b}_J}(x_{\epsilon J})=a_{1J}-b_{2J}
	=\sqrt{J}\sigma(a_{1J}^0-b_{2J}^0)
	\end{displaymath}
	after passing to the normalized coefficients. For the  
	expected surplus of $I_{{\bf a}_J{\bf b}_J}$
	we need
	the expected net reinsurance surplus which from~(\ref{e211}) is 
	\begin{displaymath}
	\pi(I_{{\bf a}_J{\bf b}_J})-E\{I_{{\bf a}_J{\bf b}_J}(X)\}
	=\int_{a_{1J}}^{x_{\epsilon J}} K\{F_J(x)\}dx
	+\int_{0}^{b_{2 J}} K\{F_J(x)\}dx
	\end{displaymath}
	or after substituting $t=(x-J\xi)/(\sqrt{J}\sigma)$ in the integral
	\begin{displaymath}
	\pi(I_{{\bf a}_J{\bf b}_J})-E\{I_{{\bf a}_J{\bf b}_J}(X)\}
	=\sqrt{J}\sigma\left(\int_{a_{1J}^0}^{x^0_{\epsilon J}} K\{F_J^0(t)\}dt
	+\int_{x_{J}^0}^{b^0_{2 J}} K\{F_J^0(t)\}dt\right)
	\end{displaymath}
	so that the expected surplus for the insurer becomes
	\begin{equation}
	{\cal G}_J(I_{{\bf a}_J{\bf b}_J})=J\xi\gamma -
	\sqrt{J}\sigma\left(\int_{a_{1J}^0}^{x^0_{\epsilon J}} K\{F_J^0(t)\}dt
	+\int_{x_{J}^0}^{b^0_{2 J}} K\{F_J^0(t)\}dt+
	\beta(a_{1J}^0-b_{2J}^0), 
	\right)
	\label{a4}
	\end{equation}
	and
	\begin{equation}
	{\cal C}_J(I_{{\bf a}_J{\bf b}_J})=\frac{\sqrt{J}\sigma(a_{1J}^0-b_{2J}^0)}
	{{\cal G}_J(I_{{\bf a}_J,{\bf b}_J})}.
	\label{a5}
	\end{equation}
	Elementary differentiation yields
	\begin{displaymath}
	\frac{\partial{\cal C}_J(I_{{\bf a}_J{\bf b}_J})}{\partial b_{2J}^0}
	=\frac{-\sqrt{J}\sigma}{{\cal G}_J(I_{{\bf a}_J{\bf b}_J})}-\frac{\sqrt{J}\sigma(a_{1J}^0-b_{2J}^0)}
	{{\cal G}_J(I_{{\bf a}_J{\bf b}_J})^2}
	\frac{\partial {\cal G}_J(I_{{\bf a}_J,{\bf b}_J})}{\partial b_{2J}^0}
	\end{displaymath}
	with
	\begin{displaymath}
	\frac{\partial {\cal G}_J(I_{{\bf a}_J{\bf b}_J})}{\partial b_{2J}^0}
	=-\sqrt{J}\sigma(K\{F_J^0(b_{2J}^0)\}-\beta).
	\end{displaymath}
	After some straightforward calculations 
	\begin{displaymath}
	\frac{\partial{\cal C}_J(I_{{\bf a}_J{\bf b}_J})}{\partial b_{2J}^0}=
\sqrt{J}\sigma
	\frac{H_J(a_{1J}^0,b_{2J}^0)}
	{{\cal G}_J(I_{{\bf a}_J,{\bf b}_J})^2}
	\end{displaymath}
	where
	\begin{displaymath}
	H_J(a_{1J}^0,b_{2J}^0)=-J\xi\gamma+\sqrt{J}\sigma\left(
	\int_{a_{1J}^0}^{x^0_{\epsilon J}} K\{F_J^0(t)\}dt
	+\int_{x_{J}^0}^{b^0_{2 J}} K\{F_J^0(t)\}dt+
	(a_{1J}^0-b_{2J}^0)K\{F_J^0(b_{2J}^0)\}\right).
	\end{displaymath}
	Whether ${\cal C}_J(I_{{\bf a}_J{\bf b}_J})$ goes up or down with
	$b_{2J}^0$
	is determined by this function
	which can be examined through
	\begin{equation}
	\frac{H_J(a_{1J}^0,b_{2J}^0)}{J\xi}
=-\gamma+A_J(a_{1J}^0,b_{2J}^0)+B_J(b_{2J}^0)
        \label{a5a}
	\end{equation}
	where since $x_J^0=-\sqrt{J}\xi/\sigma$
	\begin{equation}
	A_J(a_{1J}^0,b_{2J}^0)=-\frac{1}{x_J^0}
	\left(\int_{a_{1J}^0}^{x^0_{\epsilon J}} K\{F_J^0(t)\}dt
	+a_{1J}^0K\{F_J^0(b_{2J}^0)\}\right),
        \label{a5b}
	\end{equation}
	\begin{equation}
	B_J(b_{2J}^0)=-\frac{1}{x_J^0}
	\left(\int_{x_J^0}^{b^0_{2 J}} K\{F_J^0(t)\}dt
	-b_{2J}^0K\{F_J^0(b_{2J}^0)\}\right).
        \label{a5c}
	\end{equation}
	By Lemma \ref{lemmaA1} with $a=a_{1J}^0$, $x=x_{\epsilon J}^0$
        and $b=b_{1J}^0$
	\begin{displaymath}
	\sup|A_J(a_{1J}^0,b_{2J}^0)|\leq -\frac{x_{\epsilon J}^0}{x_J^0}\left(
	\int_0^1|W(u)-1|du+K\{F(b_{2J}^0)\}\right)
	\end{displaymath}
	where $\sup$ is over $b_{2J}^0\leq a_{1J}^0\leq x_{\epsilon J}^0$.
	Since $x_{\epsilon J}$, $K\{F(b_{2J}^0)\}$ and the integral
	on the 
	right are bounded, 
	$\sup|A_J(a_{1J}^0,b_{2J}^0)|\rightarrow 0$.
	To deal with the other quantity note that
	\begin{displaymath}
	B_J(b_{2J}^0)-K(0)
	=-\frac{1}{x_J^0}
	\left(\int_{x_J^0}^{b^0_{2 J}} K\{F_J^0(t)\}dt
	-b_{2J}^0K\{F_J^0(b_{2J}^0)\}+x_J^0K(0)\right)
	\end{displaymath}
        \begin{displaymath}
	\hspace*{2.7cm}=-\frac{1}{x_J^0}
	\left(\int_{x_J^0}^{b^0_{2 J}} [K\{F_J^0(t)\}-K(0)]dt
	+b_{2J}^0[K(0)-K\{F_J^0(b_{2J}^0)\}]\right).
	\end{displaymath}
        Moreover, from~(\ref{e210})
        \begin{displaymath}
        \int_{x_J^0}^{b^0_{2 J}} [K\{F_J^0(t)\}-K(0)]dt
         =-\int_{x_J^0}^{b^0_{2 J}} \int_0^{F_J^0(t)}(W(u)-1)dudt
	\end{displaymath}
        which becomes after changing the order of integration
        \begin{displaymath}
         =-\int_0^{F_J^0(b_{2J}^0)}\{b_{2J}^0-(F_J^0)^{-1}(u)\}(W(u)-1)du
	\end{displaymath}
        \begin{displaymath}
         =-b_{2J}^0[K(0)-K\{F_J^0(b_{2J}^0)\}]
        +\int_0^{F_J^0(b_{2J}^0)}(F_J^0)^{-1}(u)(W(u)-1)du
	\end{displaymath}
        so that
	\begin{displaymath}
	B_J(b_{2J}^0)-K(0)
	=-\frac{1}{x_J^0}
        \int_0^{F_J^0(b_{2J}^0)}(F_J^0)^{-1}(u)(W(u)-1)du.
	\end{displaymath}
        But it follows from this that
	\begin{equation}
	|B_J(b_{2J}^0)-K(0)|\leq
	\frac{|b_{2J}^0|}{|x_J^0|} F_J^0(b_{2J}^0) s_W
   	\label{a8a} 
	\end{equation}
        where
	\begin{displaymath}
   	s_W=\sup_{0<u<1-\epsilon}|W(u)-1|<\infty. 
	\end{displaymath}
        We have to show that the right hand side of~(\ref{a8a})
$\rightarrow 0$ 
as $J\rightarrow \infty$ uniformly in $b_{2J}^0$ 
        when $-\sqrt{J}\xi/\sigma=x_J^0\leq b_{2J}^0\leq x_{\epsilon J}^0$.
	Let $\eta>0$ and note that $F_J^0(x_{1-\eta,J}^0)=\eta$
and recall that $x_{1-\eta,J}\rightarrow \phi_{1-\eta}$ when $J\rightarrow \infty$.
        It follows that there exists a $J_\eta$ so that if $J>J_\eta$, then
        $|b_{2J}^0|\leq |x_{J}^0|$ when
        $x_J^0\leq b_{2J}^0\leq x_{1-\eta,J}^0$ 
and under this condition
	\begin{displaymath}
	|B_J(b_{2J}^0)-K(0)|
	\leq \eta s_W.
	\end{displaymath}
	In the opposite case when
$x_{1-\eta,J}^0< b_{2J}^0\leq x_{\epsilon J}^0$ the interval is bounded as 
$J\rightarrow \infty$, and $x_J^0$ in the denominator in~(\ref{a8a})
implies that $|B_J(b_{2J}^0)-K(0)|\rightarrow 0$ uniformly in this interval too	
	so that
	$\sup|B_J(b_{2J}^0)-K(0)| \rightarrow 0$ where the sup is over
        $x_J^0\leq b_{2J}^0\leq x_{\epsilon J}^0$.
	Finally from~(\ref{a5a})
	\begin{displaymath}
	\sup|H_J(a_{1J}^0,b_{2J}^0)/(J\xi)-(K(0)-\gamma)|\rightarrow 0
	\end{displaymath}
	where the sup is over all
	$a_{1J}^0$ and $b_{2J}^0$ 
	for which $x_J^0\leq b_{2J}^0\leq a_{1J}^0\leq x_{\epsilon J}$. 
	But since $K(0)>\gamma$, 
        this uniform bound 
	establishes 	for sufficiently large $J$
	that $H_J(a_{1J}^0,b_{2J}^0)>0$ for all $a_{1J}^0$ and $b_{2J}^0$
	which in turn implies that 
	${\cal C}_J(I_{{\bf a}_J{\bf b}_J})$ for such $J$ 
	is an increasing function of $b_{2J}^0$ for all $a_{1J}^0$  so that 
	the optimum is to remove the $b$-layer completely.
\end{proof}
We need still another lemma which utilizes that $F_J^0(x)\rightarrow \Phi(x)$
uniformly in $x$ as $J\rightarrow \infty$ where $\Phi(x)$ is the standard 
Gaussian distribution 
function; consult \cite{hall1992} (for example) for this result.
\\
\begin{lemma}\label{lemmaA4}
	Let $x_{\delta J}^0$ be the 
$1-\delta$ percentile
of $X_J^0$  which satisfies
$K\{F_J^0(x_{\delta J}^0)\}=\gamma$  and let $z_J^0$ 
 be a sequence so that   $K\{F_J^0(z_J^0)\}\rightarrow\gamma$
as $J\rightarrow \infty$. Then $z_J^0-x_{\delta J}^0\rightarrow 0$.
\end{lemma}
\begin{proof}
        Under Condition 1 there is a unique $\delta$ between $0$ and $1$  
so that $K(1-\delta)=\gamma$ which implies that
$F_J^0(z_J^0)\rightarrow 1-\delta=F_J^0(x_{\delta J}^0)$ so that
$F_J^0(z_J^0)-F_J^0(x_{\delta J}^0)\rightarrow 0$.
But since $F_J^0(x)\rightarrow \Phi(x)$ uniformly this cannot occur unless
$z_J^0-x_{\delta J}^0\rightarrow 0$.
\end{proof}
{\bf Finalizing the argument}
Proposition \ref{prop2} established that the
search for the optimal
reinsurance function can be carried out within  the two-layer class
$I_{{\bf a}_J{\bf b}_J}$ with $b_{1J}=0$
and $a_{2J}=x_{\epsilon J}$, and for large portfolios Lemma \ref{lemmaA3} 
further reduced the candidates  
to the one-layer sub-class $I_{{\bf a}_J}$ with $a_{2J}=x_{\epsilon J}$.
The coefficient $a_{1J}$ with its normalized version $a_{1J}^0$ 
is then the only remaining unknown to optimize over,
and 
it is convenient to simplify notation so that 
${\cal C}_J(a_{1J}^0)={\cal C}_J(I_{{\bf a}_J})$
and ${\cal G}_J(a_{1J}^0)={\cal G}_J(I_{{\bf a}_J})$ (the same convention is used
everywhere below).
\\\\
Value at Risk is now simply $a_{1J}=J\xi+\sqrt{J}\sigma a_{1J}^0$
so that the risk over surplus ratio becomes
\begin{displaymath}
{\cal C}_J(a_{1J}^0)=\frac{J\xi+\sqrt{J}\sigma a_{1J}^0}
{ {\cal G}_J(a_{1J}^0)}
\end{displaymath}
where after removing the $\bf b$-layer in~(\ref{a4})
\begin{equation}
{\cal G}_J(a_{1J}^0)=J\xi\gamma -
\sqrt{J}\sigma\int_{a_{1J}^0}^{x^0_{\epsilon J}} K\{F_J^0(t)\}dt
-\beta(J\xi+\sqrt{J}\sigma a_{1J}^0).
\label{a6}
\end{equation}
Differentiation yields
\begin{displaymath}
\frac{\partial {\cal C}_J(a_{1J}^0)}{\partial a_{1J}^0}
=\frac{\sqrt{J}\sigma}{ G_J(a_{1J}^0)}-
\frac{J\xi+\sqrt{J}\sigma a_{1J}^0}{G_J(a_{1J}^0)^2}
\sqrt{J}\sigma (K\{F_J^0(a_{1J}^0)\}-\beta),
\end{displaymath}
or after some straightforward calculations
\begin{displaymath}
\frac{\partial {\cal C}_J(a_{1J}^0)}{\partial a_{1J}^0}
=\sqrt{J}\sigma\frac{J\xi[\gamma-K\{F_J^0(a_{1J}^0)\}]
	-\sqrt{J}\sigma H_J(a_{1J}^0)}
{ {\cal G}_J(a_{1J}^0)^2}.
\end{displaymath}
where
\begin{equation}
H_J(a_{1J}^0)=\int_{a_{1J}^0}^{x_{\epsilon J}^0}
K\{F_J^0(t)\}dt+a_{1J}^0 K\{F_J^0(a_{1J}^0)\}.
\label{a7}
\end{equation}
Hence
\begin{equation}
\frac{\partial {\cal C}_J(a_{1J}^0)}{\partial a_{1J}^0}=
J^{3/2}\sigma\frac{\xi(\gamma-K\{F_J^0(a_{1J}^0)\})-\sigma H_J(a_{1J}^0)/\sqrt{J}}
{{\cal G}_J(a_{1J}^0)^2}.
\label{a8}
\end{equation}
By Lemma \ref{lemmaA1}, with $x=x_{\epsilon J}^0, a=a_{1J}^0=b$,
\begin{displaymath}
|H_J(a_{1J}^0)|\leq |x_{\epsilon J}^0|\Big(
\int_0^1|W(u)-1|du +K\{F_J^0(a_{1J}^0)\}\Big)
\end{displaymath}
and $\sup|H_J(a_{1J}^0)|/\sqrt{J}\rightarrow 0$
with the sup taken over all
$a_{1J}^0\leq x_{\epsilon J}^0$.
\\\\
Let $\tilde{a}_{1J}^0$ be the 
value minimizing
${\cal C}_J(a_{1J}^0)$ which from~(\ref{a8})
must satisfy that
\begin{displaymath}
K\{F_J^0(\tilde{a}_{1J}^0)\}]
+(\sigma/\xi) H_J(\tilde{a}_{1J}^0)/\sqrt{J}=\gamma,
\end{displaymath}
and since  $H_J(a_{1J}^0)/\sqrt{J}\rightarrow 0$ uniformly in $\tilde{a}_{1J}^0$,
it follows that $K\{F_J^0(\tilde{a}_{1J}^0)\}\rightarrow \gamma$.
But $K\{F_J^0(x_{\delta J}^0)\}=\gamma$, and by Lemma \ref{lemmaA4}
this can not occur unless
$\tilde{a}_{1J}^0-x_{\delta J}^0\rightarrow 0$. Hence $\hat{a}_{1J}^0 = \hat{x}_{\delta J}^0$ so that 
${\cal C}_J(\tilde{a}_{1J}^0))-{\cal C}_J(\hat{a}_{1J}^0) \rightarrow 0$
as well, and there exists for any $\eta >0$ some $J_\eta$
so that for any $a_{1J}^0$
\begin{displaymath}
{\cal C}_J(a_{1J}^0)\geq {\cal C}_J(\tilde{a}_{1J}^0)\geq
{\cal C}_J(\hat{a}_{1J}^0)-\eta
\end{displaymath}
which completes the proof of Proposition \ref{prop4.1}. 

\subsection{ Proposition \ref{prop4.2}}\label{appenA2}
{\em Part 1} We need asymptotic expressions 
for the first and second derivative of
${\cal C}_J(a_{1J}^0)$ at $a_{1J}^0=x_{\delta J}^0$. In~(\ref{a8}) the 
first term in the numerator then vanishes 
since $K\{F_J^0(x_{\delta J}^0)\}=\gamma$
so that
\begin{equation}
\frac{\partial {\cal C}_J(x_{\delta J}^0)}{\partial a_{1J}^0}=
-J\frac{\sigma^2 H_J(x_{\delta J}^0)}
{{\cal G}_J(x_{\delta J}^0)^2},
\label{a11a}
\end{equation}
whereas from~(\ref{a6}) 
\begin{displaymath}
{\cal G}_J(x_{\delta J}^0)=J\xi(\gamma -\beta)+o(J),
\end{displaymath}
and from~(\ref{a7}) since $F_J^0(x_{\delta J})=1-\delta$
\begin{displaymath}
H_J(x_{\delta J}^0)=\int_{x_{\delta J}^0}^{x^0_{\epsilon J}}
K\{F_J^0(t)\}dt+x_{\delta J}K(1-\delta).
\end{displaymath}
But $F_J^0(x)\rightarrow \Phi(x)$, $x^0_{\delta J}\rightarrow \phi_\delta$ 
as $J\rightarrow \infty$ and
$x^0_{\epsilon J}\rightarrow \phi_\epsilon$ so that
\begin{displaymath}
H_J(x_{\delta J}^0)=\int_{\phi_\delta}^{\phi_{\epsilon}}
K\{\Phi(t)\}dt+\phi_{\delta }K(1-\delta) + o(1),
\end{displaymath}
and when the expressions for ${\cal G}_J(x_{\delta J}^0)$
and $H_J(x_{\delta J}^0)$ are inserted in~(\ref{a11a}), it emerges that
\begin{equation}
\frac{\partial {\cal C}_J(x_{\delta J}^0)}{\partial a_{1J}^0}=-
\frac{B_1}{J\xi^2(\gamma-\beta)^2}+o(1/J)
\quad\mbox{with}\quad B_1=\sigma^2\left(
\int_{\phi_\delta}^{\phi_{\epsilon}}
K\{\Phi(t)\}dt+\phi_{\delta }K(1-\delta)
\right)
\label{a9}
\end{equation}
The
second derivative can be calculated 
from~(\ref{a8}) and has a complicated expression. However, only the leading 
term is required, and this boils down to 
\begin{displaymath}
\frac{\partial^2C_J(x_{\delta J}^0)}
{\partial (a_{1J}^0)^2}
=-\frac{K'\{F_{J}^0(x_{\delta J})\}f_{J}^0(x_{\delta J})}
{\xi(\gamma-\beta)^2}\frac{1}{\sqrt{J}}+o(1/\sqrt{J})
\end{displaymath}
where $f_{J}^0(x)=d F_J^0(x)/dx$.
But the central limit theorem on density form 
yields $f_J^0(x)\rightarrow \Phi'(x)$ and as above
\begin{equation}
\frac{\partial^2C_J(x_{\delta J}^0)}
{\partial (a_{1J}^0)^2}
=-\frac{K'(1-\delta)\Phi'(\phi_\delta)}
{\xi(\gamma-\beta)^2}\frac{1}{\sqrt{J}}+o(1/\sqrt{J}).
\label{a10}
\end{equation}
{\em Part 2} We seek the asymptotic degradation when using
the normalizing coefficients 
$a_{1J}^0=x_{\delta J}$ instead of the optimal 
$\tilde{a}_{1J}^0$ with  in both cases
$a_{2J}^0=x_{\epsilon J}^0$ as upper limit.
An elementary one-variable Taylor argument  yields
\begin{displaymath}
{\cal C}_J(x_{\delta J}^0)={\cal C}_J(\tilde{a}_{1J}^0)
+\frac{1}{2}\frac{\partial^2{\cal C}_J(v_{1J})}
{\partial (a_{1J}^0)^2}
(x_{\delta J}^0-\tilde{a}_{1J}^0)^2 
\end{displaymath}
with $v_{1J}$ between 
$\tilde{a}_{1J}^0$ and $x_{\delta J}^0$. Note that
the linear term has vanished since $\tilde{a}_{1J}^0$ is minimizing. 
The degradation 
${\cal D}_J(x_{\delta J}^0)={\cal C}_J(x_{\delta J}^0)-{\cal C}_J(\tilde{a}_{1J}^0)$
in using
$a_{1J}^0=x_{\delta J}^0$ instead of $a_{1J}^0=\tilde{a}_{1J}^0$
is then
\begin{equation}
{\cal D}_J(x_{\delta J}^0)=
\frac{1}{2}\frac{\partial^2{\cal C}_J(v_{1J})}
{\partial (a_{1J}^0)^2}
(x_{\delta J}^0-\tilde{a}_{1J}^0)^2
\label{a11}
\end{equation}
where an assessment of $x_{\delta J}^0-\tilde{a}_{1J}^0$
is needed. The mean value theorem implies that
\begin{equation}
\frac{\partial {\cal C}_J(x_{\delta J}^0)}{\partial a_{1J}^0}=
\frac{\partial C_J(\tilde{a}_{1J}^0)}{\partial a_{1J}^0}
+\frac{\partial^2 C_J(v_{2J})}
{\partial (a_{1J}^0)^2}
(x_{\delta J}^0-\tilde{a}_{1J}^0) 
=\frac{\partial^2 C_J(v_{2J})}
{\partial (a_{1J}^0)^2}
(x_{\delta J}^0-\tilde{a}_{1J}^0)
\label{a12} 
\end{equation}
with $v_{2J}$ between 
$\tilde{a}_{1J}^0$ and $x_{\delta J}^0$. But the second order derivatives in~(\ref{a11}) and~(\ref{a12})
both tends to $\partial^2{\cal C}_J(x_{\delta J}^0)/
\partial (a_{1J}^0)^2$
since $x_{\delta J}^0-\tilde{a}_{1J}^0\rightarrow 0$ and 
both $v_{1J}$ and $v_{2J}$
are squeezed in between. By combining~(\ref{a11}) and~(\ref{a12})
it follows that
\begin{equation}
{\cal D}_J(x_{\delta J}^0)=
\frac{1}{2}\left(\frac{\partial {\cal C}_J(x_{\delta J}^0)}
{\partial a_{1J}^0}\right)^2
\left(\frac{\partial^2{\cal C}_J(x_{\delta J})}
{\partial (a_{1J}^0)^2}\right)^{-1}+o(1/J^{3/2})
\label{a13}
\end{equation}
where the error term is a consequence of those in~(\ref{a9}) 
and~(\ref{a10}), and inserting those leads after a straightforward 
calculation to 
\begin{displaymath}
{\cal D}_J(x_{\delta J}^0)=-\frac{1}{J^{3/2}}
\frac{B_1^2/\{\xi^3(\gamma-\beta)^2\}}{K'(1-\delta)\Phi'(\phi_\delta)}
+o(1/J^{3/2})
\end{displaymath}
as claimed in Proposition \ref{prop4.2}.
\subsection{Proposition \ref{prop4.3}}\label{appenA3}
The normalized coefficients 
when the reinsurance function are using the Gaussian percentiles
in~(\ref{e418}) are
${a}_{1J}^0=\phi_\delta$ and
${a}_{2J}^0=\phi_\epsilon$, and the upper
limit deviates from the exact one $x^0_{\epsilon J}$ which  changes  
things considerably. Value at Risk is now
\begin{displaymath}
a_{1J}+(x_{\epsilon J}-a_{2J})_+
=J\xi+\sqrt{J}\sigma(
\phi_\delta+(x_{\epsilon J}^0-\phi_\epsilon)_+).
\end{displaymath}
after inserting $a_{1J}=J\xi+\sqrt{J}\sigma\phi_\delta$,
$a_{2J}=J\xi+\sqrt{J}\sigma\phi_\epsilon$ and
$x_{\epsilon J}=J\xi+\sqrt{J}\sigma x_{\epsilon J}^0$. Recall that we are assuming 
$x^0_{\epsilon J}>\phi_\epsilon$, and the 
degradation in using ${a}_{1J}^0=\phi_\delta$ and
${a}_{2J}^0=\phi_\epsilon$ instead of the optimal
${a}_{1J}^0=\tilde{a}_{1J}^0$ and
${a}_{2J}^0=x_{\epsilon J}^0$ 
then becomes
\begin{displaymath}
{\cal D}_J(\phi_\delta,\phi_\epsilon)=
\frac{J\xi+\sqrt{J}\sigma(
	\phi_\delta+x_{\epsilon J}^0-\phi_\epsilon)}
{{\cal G}_J(\phi_\delta,\phi_\epsilon)}
-\frac{J\xi+\sqrt{J}\sigma\tilde{a}_{1J}^0}
{{\cal G}_J(\tilde{a}_{1J}^0,x_{\epsilon J}^0)}
\end{displaymath}
with
${\cal G}_J(\phi_\delta,\phi_\epsilon)$ and ${\cal G}_J(\tilde{a}_{1J}^0,x_{\epsilon J}^0)$
the expected surplus terms. 
This may be rewritten 
\begin{equation}
{\cal D}_J(\phi_\delta,\phi_\epsilon)=A_{1J}+A_{2J}+A_{3J}
\label{a14}
\end{equation}
where
\begin{equation}
A_{1J}=\sqrt{J}\sigma\frac{x_{\epsilon J}^0-\phi_\epsilon}
{{\cal G}_J(\phi_\delta,\phi_\epsilon)},
\label{a15}
\end{equation}
\begin{equation}
A_{2J}=\frac{J\xi+\sqrt{J}\sigma\phi_\delta}
{{\cal G}_J(\phi_\delta,x_{\epsilon J}^0)}
-\frac{J\xi+\sqrt{J}\sigma\tilde{a}_{1J}^0}
{{\cal G}_J(\tilde{a}_{1J}^0,x_{\epsilon J}^0)}.
\label{a16}
\end{equation}
\begin{equation}
A_{3J}=\frac{J\xi+\sqrt{J}\sigma\phi_\delta}
{{\cal G}_J(\phi_\delta,\phi_\epsilon)}
-\frac{J\xi+\sqrt{J}\sigma\phi_\delta}
{{\cal G}_J(\phi_\delta,x_{\epsilon J}^0)}.
\label{a17}
\end{equation}
The second of these terms represents degradation 
due to the difference between $\tilde{a}_{1J}^0$ and 
$\phi_\delta$ with the upper limit the optimal $x_{\epsilon J}^0$ 
and is on the argument that lead to
Proposition \ref{prop4.2} of order $o(1/J)$ whereas 
$A_{3J}$ must be examined further.
Note that
\begin{equation}
A_{3J}=\frac{J\xi+\sqrt{J}\sigma\phi_\delta}
{{\cal G}_J(\phi_\delta,\phi_\epsilon)
{\cal G}_J(\phi_\delta,x_{\epsilon J}^0)}
\{{\cal G}_J(\phi_\delta,x_{\epsilon J}^0)
-{\cal G}_J(\phi_\delta,\phi_{\epsilon})
\}.
\label{a18}
\end{equation}
where
\begin{displaymath}
{\cal G}_J(\phi_\delta, \phi_{\epsilon })=
J\xi\gamma -
\sqrt{J}\sigma\int_{\phi_\delta}^{\phi_{\epsilon J}} K\{F_J^0(t)\}dt
-\beta\{J\xi+\sqrt{J}\sigma(\phi_\delta+x_{\epsilon J}^0-\phi_\epsilon) \},
\end{displaymath}
\begin{displaymath}
{\cal G}_J(\phi_\delta, x_{\epsilon J}^0)=
J\xi\gamma -
\sqrt{J}\sigma\int_{\phi_\delta}^{x^0_{\epsilon J}} K\{F_J^0(t)\}dt
-\beta\{J\xi+\sqrt{J}\sigma\phi_\delta\}.
\end{displaymath}
Hence
\begin{displaymath}
{\cal G}_J(\phi_\delta,x_{\epsilon J}^0)
-{\cal G}_J(\phi_\delta,\phi_{\epsilon})
=\sqrt{J}\sigma\left(-\int_{\phi_\delta}^{x_{\epsilon J}^0} K\{F_J^0(t)\}dt
+\int_{\phi_\delta}^{\phi_{\epsilon }} K\{F_J^0(t)\}dt
+\beta(x_{\epsilon J}^0-\phi_\epsilon)\right)
\end{displaymath}
\begin{displaymath}
\hspace*{4.10cm}=
\sqrt{J}\sigma(-K\{F_J^0(\phi_{\epsilon })\}+\beta)(x_{\epsilon J}^0-\phi_\epsilon)
+o(1)
\end{displaymath}
after Taylor's formula has been applied to the difference between the integrals.
But $F_J^0(\phi_\epsilon)\rightarrow 1-\epsilon$
as $J\rightarrow \infty$ so that~(\ref{a18}) after some 
straightforward calculations becomes
\begin{displaymath}
A_{3J}=\frac{1}{\sqrt{J}}\frac{\sigma(-K(1-\epsilon)+\beta)}{\xi(\gamma-\beta)^2}
(x_{\epsilon J}^0-\phi_\epsilon)+o(1/J).
\end{displaymath}
This is a quantity of the same order of magnitude as $A_{1J}$ 
which by~(\ref{a15}) can be rewritten
\begin{displaymath}
A_{1J}=\frac{1}{\sqrt{J}}\frac{\sigma}{\xi(\gamma-\beta)}
(x_{\epsilon J}^0-\phi_\epsilon)
+o(1/J).
\end{displaymath}
These two must be added whereas $A_{2J}$ is of smaller order and can be ignored
so that~(\ref{a14}) is
${\cal D}_J(\phi_\delta,\phi_\epsilon)=A_{1J}+A_{3J}+o(1/J)$ or after a little calculation
\begin{displaymath}
{\cal D}_J(\phi_\delta,\phi_\epsilon)=\frac{1}{\sqrt{J}}\zeta_2
(x_{\epsilon J}^0-\phi_\epsilon)+o(1/J)
\quad\mbox{where}\quad
\zeta_2=\frac{\sigma\{\gamma-K(1-\epsilon)\}}{\xi(\gamma-\beta)^2}
\end{displaymath}
as claimed in Proposition \ref{prop4.3}.

\section{The simulation experiment in 
Section \ref{sec:opti}}\label{appenB}
Section \ref{subsec:num1} required the calculations of $K\{F(x)\}$
which is the main part of $\psi_\lambda(x)$ in~(\ref{e213}).
This requires a joint model for the uniform pair $(U,V)$ underlying $(X,Z)$.
We have used the Clayton Copula 
\begin{displaymath}
C(u,v)=(u^{-\theta}+v^{-\theta}-1)^{-1/\theta},
\hspace*{1cm}0<u,v<1
\end{displaymath}
with $\theta=10$, and $U$ and $V$ are then passed on
through $X=F^{-1}(U)$ and $Z=G^{-1}(V)$ where $F(x)$ and $G(z)$ 
are the distribution functions
of $X$ and $Z$.
From the definition of $K(u)$ in~(\ref{e210}) 
and $W(u)$ in~(\ref{e29}) right it follows that
\begin{displaymath}
K\{F(x)\}=E(\{M(Z)-1\}{\cal I}_{U>F(x)})
\end{displaymath}
which can be approximated by Monte Carlo through the following steps:
\\\\
\hspace*{1cm}1. Generate simulations $X_1^*,\dots,X^*_m$ of $X$.\\
\hspace*{1cm}2. Approximate $F(x)$ through the kernel density estimate
\begin{displaymath}
\hspace*{1cm}F^*(x)=\frac{1}{m}\sum_{i=1}^m\Phi\{(x-X^*_i)/(s^*h)\}
\end{displaymath}
\hspace*{1.5cm}with $\Phi(x)$ the 
Gaussian integral, $s^*$ the standard deviation of
$X_1^*,\dots,X^*_m$ and $h=0.2$.\\
\hspace*{1cm}3. Calculate $U_i^*=F^*(X_i^*)$, $i=1,\dots,m$.\\
\hspace*{1cm}4. Generate 
$Y^*_i\sim \mbox{ uniform}$, $i=1,\dots,m$.\\
\hspace*{1cm}5. Calculate 
$V_i^*=\{1+(U_i^*)^{-\theta}(Y_i^{-\theta/(1+\theta)})-1\}^{-1/\theta}$, $i=1,\dots,m$.\\
\hspace*{1cm}6. Calculate $Z^*_i=G^{-1}(V^*_i)$, $i=1,\dots,m$.
\\\\
The approximations of $K\{F(x)\}$ then becomes
\begin{displaymath}
\hspace*{1cm}K^*\{F^*(x)\}=\frac{1}{m}\sum_{i=1}^m M(Z_i^*){\cal I}_{U^*_i>F^*(x)}.
\end{displaymath}

Note that all Monte Carlo simulations have been $^*$-marked.
The first step is carried out by an ordinary program for simulating 
portfolio losses whereas steps $4$ and $5$ 
is one of the ways the Clayton copula can be simulated;
consult p.208 in \cite{bolviken_2014}. The final step $6$ makes use of the 
percentile function $G^{-1}(z)$ which is available for all 
standard distributions.

\end{document}